\documentclass[journal]{IEEEtran}
\IEEEoverridecommandlockouts
\usepackage{amsmath,amsfonts}
\usepackage{array}
\usepackage{textcomp}
\usepackage{stfloats}
\usepackage{url}
\usepackage{verbatim}
\usepackage{graphicx}
\usepackage{cite}
\usepackage[noend]{algpseudocode}
\usepackage[linesnumbered,ruled]{algorithm2e}
\usepackage{upgreek}
\usepackage{subfigure}
\usepackage{color}
\def\BibTeX{{\rm B\kern-.05em{\sc i\kern-.025em b}\kern-.08em
    T\kern-.1667em\lower.7ex\hbox{E}\kern-.125emX}}
    
\usepackage{amsthm}
\usepackage{amssymb}
\usepackage{makecell}
\usepackage{enumitem}
\newtheorem{theorem}{Theorem}

\theoremstyle{remark}
\theoremstyle{proposition}

\newtheorem{proposition}{Proposition}
\theoremstyle{Assumption}

\theoremstyle{Lemma}

\begin{document}

\graphicspath{{figures/}}

\title{TrimCaching: Parameter-sharing AI Model Caching in Wireless Edge Networks\\
\thanks{Guanqiao Qu, Zheng Lin, Xianhao Chen, and Kaibin Huang are with the Department of Electrical and Electronic Engineering, University of Hong Kong, Pok Fu Lam, Hong Kong SAR, China. (e-mail: gqqu@eee.hku.hk; linzheng@eee.hku.hk; xchen@eee.hku.hk; huangkb@eee.hku.hk). Fangming Liu is with Peng Cheng Laboratory, Shenzhen, China, and also with the Huazhong University of Science and Technology, Wuhan, Hubei, China. (e-mail: fangminghk@gmail.com). \textit{(Corresponding author: Xianhao Chen.) This paper has been accepted by ICDCS 2024.}
}
}
\author{Guanqiao Qu,~\IEEEmembership{Graduate Student Member,~IEEE}, Zheng Lin,~\IEEEmembership{Graduate Student Member,~IEEE}, \\ Fangming Liu,~\IEEEmembership{Senior Member,~IEEE}, Xianhao Chen,~\IEEEmembership{Member,~IEEE}, Kaibin Huang,~\IEEEmembership{Fellow,~IEEE}
}

\maketitle

\begin{abstract}
Next-generation mobile networks are expected to facilitate fast AI model downloading to end users. By caching models on edge servers, mobile networks can deliver models to end users with low latency, resulting in a paradigm called \textit{edge model caching}. In this paper, we develop a novel model placement scheme, called parame\underline{t}e\underline{r}-shar\underline{i}ng \underline{m}odel caching (TrimCaching). TrimCaching exploits the key observation that a wide range of AI models, such as convolutional neural networks or large language models, can share a significant proportion of parameter blocks containing reusable knowledge, thereby improving storage efficiency. To this end, we formulate a parameter-sharing model placement problem to maximize the cache hit ratio in multi-edge wireless networks by balancing the fundamental tradeoff between storage efficiency and service latency. We show that the formulated problem is a submodular maximization problem with submodular constraints, for which no polynomial-time approximation algorithm exists. To overcome this challenge, we study an important special case, where a small fixed number of parameter blocks are shared across models, which often holds in practice. In such a case, a polynomial-time algorithm with $\left(1-\epsilon\right)/2$-approximation guarantee is developed. Subsequently, we address the original problem for the general case by developing a greedy algorithm. Simulation results demonstrate that the proposed TrimCaching framework significantly improves the cache hit ratio compared with state-of-the-art content caching without exploiting shared parameters in AI models.
\end{abstract}

\section{Introduction}
In the era of Artificial Intelligence of Things (AIoI), there is an ongoing trend for pushing Artificial Intelligence (AI) services from the cloud to local devices due to the stringent latency and privacy requirements of diverse AI-empowered applications~\cite{lin2024efficient,lyu2023optimal}. 
For instance, autonomous driving must perceive the surrounding environments in a timely fashion~\cite{9284628,chen2024vehicle}; 
the emerging large language models (LLMs) often involve privacy-sensitive personal data for human-machine interactions~\cite{9296274,nguyen2022federated}. On the other hand, due to the sheer number of AI applications and the ever-growing size of AI models (e.g., Google's on-device LLM Gemini Nano-2 with 3.25 billion parameters), it is impractical for users to store every AI model locally. For this reason, next-generation mobile networks aim to support rapid model downloading for mobile users to realize ``on-device AI'' with better latency and privacy guarantees. By ensuring efficient model downloading, the 6th generation (6G) mobile networks enable mobile users to enjoy versatile on-device AI services with low latency without exhausting their data storage capacities.
\par

Conventional model downloading features the delivery of AI models from the remote cloud to local devices, which is infeasible for latency-sensitive services. For example, as specified by 3GPP 5G technical specifications, both autonomous vehicles and robots may need to accomplish AI model downloading within 1 second~\cite{3gpp.22.874}. Observing the large-sized models and the limited bandwidth of remote cloud centers, this kind of fast model downloading may only be achieved by placing AI models in edge networks close to users, leading to a paradigm called ``edge model caching''. However, unlike cloud centers, an edge server can only store a limited number of popular models under the storage capacities of edge servers. To enhance model downloading performance, one fundamental research problem, therefore, is model placement, e.g., how can we optimally place AI models on edge servers to maximize the cache hit ratio for model downloading under the latency requirements of various AI services? 

In this paper, we identify an AI model placement problem by observing parameter sharing among AI models. Shared parameters pervade AI models nowadays, which can be exploited to store them more efficiently at the network edge. A broad range of AI models, such as convolutional neural networks (CNNs) or LLMs, can share a considerable percentage of parameters since they can be created from the same pre-trained models~\cite{padmanabhan2022gemel}. For instance, layer freezing is a very classic transfer learning or multi-task learning method because bottom layers in deep neural networks, such as convolution layers, usually share common knowledge reusable for different downstream tasks\cite{tenser2023,zhuang2020comprehensive,Guo_2019_CVPR}. In the era of LLMs, parameter-efficient fine-tuning (PEFT) (e.g., LoRA\cite{hu2022lora}) is highly effective in fine-tuning foundation models to downstream or customized models, where the core idea is to freeze the significant proportion of parameters (e.g., more than 99\% in LoRA for LLMs) to save computing and memory resources. Today's engineers often fine-tune pre-trained models (e.g., pre-trained CNNs on the Pytorch platform or open-source foundation models such as GPT-3) on task-specific datasets to create new services. To demonstrate this phenomenon, Fig. \ref{fig_intro} shows that the inference accuracy in ResNet50 only degrades slightly as the number of shared layers grows. Even if the first 90\% of trainable layers, up to layer 97, are frozen, the average accuracy degradation is only about 4.7\% compared with the full-layer fine-tuning. This demonstrates downstream models can share a significant proportion of parameters from pre-trained models, making storage more efficient.\par

\begin{figure}[!t]
	\centering
\includegraphics[width=1.9in]{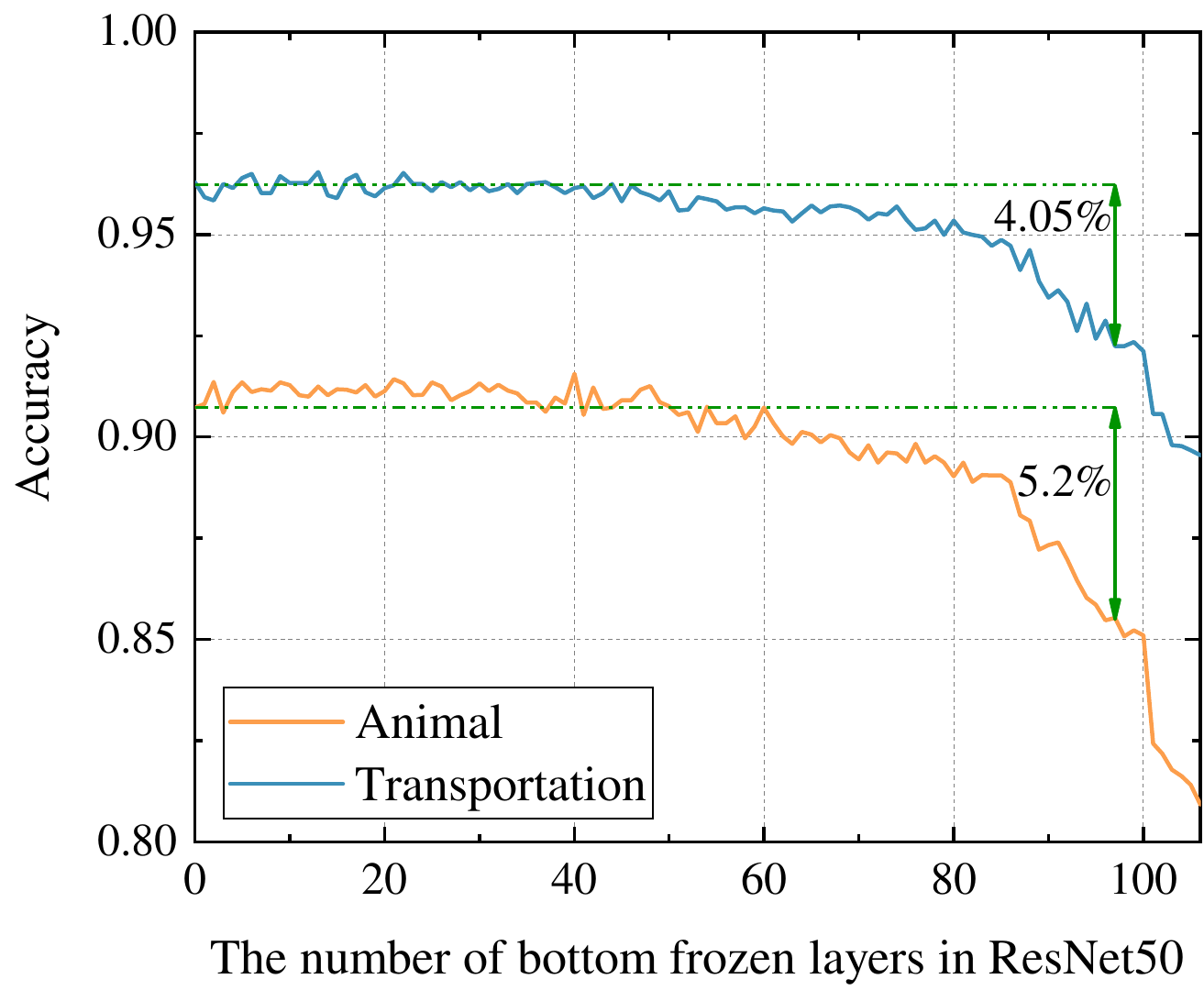}\label{fig_intro_50}
	\caption{Inference accuracy v.s the number of bottom frozen layers of fine-tuned models. Based on an original model ResNet50 \cite{he2016deep} pre-trained on CIFAR100 \cite{krizhevsky2009learning}, we fine-tune it for two downstream tasks, i.e., ``transportation'' and ``animal'', respectively. The class ``airplane", ``automobile", ``ship", and ``truck" in CIFAR10 \cite{krizhevsky2009learning} are summarized into a superclass ``transportation'', while the classes ``bird", ``cat", ``deer", ``dog", ``frog", and ``horse" are summarized into the superclass ``animal''. This implies downstream or personalized models can have a significant proportion of shared model parameters, given that fine-tuning techniques are widely adopted nowadays.}\label{fig_intro}
\end{figure}

By taking advantage of the aforementioned salient property, we will design a parame\underline{t}e\underline{r}-shar\underline{i}ng \underline{m}odel caching  (TrimCaching) framework for edge model caching. Specifically, given a set of wireless edge servers, we address the problem of placing models on edge servers to maximize the cache hit ratio for AI model downloading~\cite{6600983,7439797,8169053,7797148}. While similar problems have been studied in content placement, the shared parameters in AI model caching distinguish our scheme from traditional content placement problems for general contents (e.g., movies), resulting in a novel and very challenging optimization problem. The main contributions of this paper are summarized as follows.

\begin{enumerate}
	\item We define the parameter-sharing model caching problem. In multi-edge scenarios, we formulate the model placement problem to maximize the cache hit ratio (e.g., the ratio of downloading requests that can be served within delay requirements) under the storage capacity constraints of edge servers. 
    \item We show that the resultant problem is a submodular maximization problem with submodular constraints. After problem mapping, we conclude that there is no polynomial-time algorithm to solve it with a constant approximation ratio due to the submodular constraints resulting from shared parameter blocks.
	\item We investigate a special case of the proposed problem with a small fixed number of shared parameter blocks independent of the problem scale (e.g., the scale of the AI model library). We argue this special case often holds in reality. We develop a successive greedy and dynamic programming (DP)-based algorithm for a $\left(1-\epsilon\right)/2$-approximate solution with polynomial time complexity.
	\item We design a greedy algorithm for the original problem for the general case. Although a constant approximation guarantee cannot be achieved as alluded to earlier, we show that the algorithm is still efficient and effective through performance evaluation.
\end{enumerate}\par 
The rest of this paper is organized as follows. Section II introduces the related work. Section III elaborates on the system model and the TrimCaching framework. The problem formulation is presented in Section IV. The model placement algorithm for the special case is developed in Section V. The model placement algorithm for the general case is provided in Section VI. Section VII presents the simulation results, and Section VIII concludes the paper.\par

\section{Related Work}
Content caching in wireless edge networks has attracted significant attention since it brings contents closer to end users, leading to a field called edge caching~\cite{8291028,7565183,7565185,10068141}. Specifically, popular files can be pre-cached in wireless edge servers, such as small base stations \cite{7572146}, thereby enabling content downloading to users with low latency.
One fundamental research problem in edge caching is how to place contents on edge servers for enhancing the quality of experience (QoE) of end users, which is called the content placement problem~\cite{9126265,8422142}.
Due to the overlapping coverage of wireless edge servers, users can download files from any of the nearby edge servers with the desired contents. This class of QoE maximization problems has been identified as submodular maximization problems with knapsack constraints~\cite{6600983,7439797,8169053,7797148}. 
In these schemes, since contents are independent of each other, the storage constraints are naturally knapsack constraints. Simple greedy methods are usually effective in solving such constrained submodular maximization problems with good approximation ratios~\cite{6600983}. In our work, due to the parameter sharing across AI models, the shared parameters only need to be cached once in a server, resulting in the submodularity of storage constraints. For this reason, when adapting the aforementioned content placement strategies, e.g., greedy algorithms~\cite{6600983}, to our case, no theoretical constant approximation guarantee can be achieved.
\par


In the field of edge intelligence, caching AI models in distributed wireless edge networks can facilitate model delivery for both inference and training \cite{xu2020survey,lin2023pushing,xu2024cached,tang2023energy,lin2024split}. End users can download the required AI models from edge servers, thereby significantly reducing service latency~\cite{9522156, 10183793}. Analogous to traditional edge caching, these studies focus on optimizing model placement decisions to enhance the QoE of users subject to the storage limitations of edge servers. However, these papers have not considered parameter sharing either. In \cite{wu2024efficient}, Wu et al. propose a multi-user model downloading method by exploiting the shared parameters in AI models, but their scheme focuses on designing a multicasting scheme rather than model placement. To our best knowledge, our work is the first to define the parameter-sharing model placement problem, identify its mathematical properties, and develop the corresponding solution approaches.
\par 

\section{Parameter-sharing Model Caching Framework}
\subsection{Network Scenario}
We consider a typical scenario in wireless edge networks, as shown in Fig. \ref{fig_system_model}. Let the set $\mathcal{M}=\left\{1,2,\dots,m\dots,M\right\}$ denote the $M$ wireless edge servers (e.g., base stations). All edge servers are interconnected. There are a set ${\mathcal{K}} = \left\{1,2,\dots,k,\dots,K\right\}$ of users covered by these edge servers. The network operator provides a library of AI models ${\mathcal{I}} = \left\{1,2,\dots,i,\dots,I\right\}$, which can be downloaded by users for inference services. The request probability and the quality of service (QoS) constraint on end-to-end (E2E) latency of user $k$ for model $i$ are ${p}_{k,i}$ and $\bar{T}_{k,i}$, respectively. \par 

To ensure timely model downloading, the cloud center will push a set of models in the library to wireless edge servers in the offline stage. Consider user $k$ requests model $i$. 
Two cases of model downloading from edge networks are summarized below, where $\mathcal{M}_k$ denotes the edge servers covering user $k$.   \par 
\begin{itemize}
	\item \textbf{Downloading from associated edge servers}: User $k$ first sends the model request to its associated edge servers. If user $k$ can directly download model $i$ from any edge server $m$ in $\mathcal{M}_k$ storing this model and finish the on-device inference within $\bar{T}_{k,i}$, a cache hit occurs.
 
	\item\textbf{Downloading from non-associated edge servers}: If no edge server in $\mathcal{M}_k$ caches model $i$, the remaining edge servers will provide the model to the user. Specifically, model $i$ will be transferred from the edge server caching the model to an edge server in $\mathcal{M}_k$ for downloading to the user with the lowest latency. A cache hit occurs if the E2E latency (including edge-to-edge, edge-to-user, and on-device inference latency) meets the latency requirement $\bar{T}_{k,i}$. 
\end{itemize}\par 
During the model caching decision stage, we use the expected downloading data rate $\bar{C}_{m,k}$ from edge server $m$ to its associated user $k$ to  determine the model placements, where 
\begin{gather}\label{eq_communication}
\bar{C}_{m,k} = \bar{B}_{m,k}{\rm{log}}_2\left(1+\frac{\bar{P}_{m,k} \gamma_0  d_{m,k}^{-\alpha_0}}{n_0\bar{B}_{m,k}}\right), 
\end{gather}
where $\gamma_0$ is the antenna-related factor, $\alpha_0$ is the path loss factor, $d_{m,k}$ is the distance between user $k$ and edge server $m$, $\bar{P}_{m,k}$ and $\bar{B}_{m,k}$ are the expected transmit power and spectrum bandwidth of edge server $m$ allocated to user $k$, respectively, and $n_0$ is the power spectral density of the additional Gaussian white noise. Besides, we assume the transmission data rate between edge server $m$ and $m'$ is a constant value $C_{m,m'}$.


\subsection{Parameter-sharing Model Library}
We consider parameter sharing in the model library $\mathcal{I}$. There are $J$ parameter blocks in $\mathcal{I}$ in total. A parameter block refers to a set of parameters, which reduces problem complexity given that AI models can contain billions of parameters. A parameter block can refer to a layer in a CNN, a block in a transformer, a set of low-dimensional trainable parameters in PEFT (e.g., the LoRA technique), and even an entire backbone network, and so on, depending on how parameters are shared and non-shared. We use set $\mathcal{J}$ to represent all parameter blocks and set $\mathcal{I}_j$ to denote models containing parameter block $j$. It is worth noting that a parameter block can be exclusive to one model or shared by multiple models. For ease of presentation and schematic design, a parameter block shared by more than one model in the library is referred to as a shared parameter block; otherwise, it is called a specific parameter block.\par 

\subsection{Design Objective}
The TrimCaching framework aims to judiciously place AI models to serve as many user requests as possible by exploiting shared parameter blocks among models. For cache misses on edge servers, the downloading request can be forwarded to the cloud center for fetching the model. However, since downloading models from the cloud can be much slower, our policies aim to maximize the cache hit ratio, defined as the probability of successfully downloading the AI models from edge servers within latency constraints, which is also a common design objective in edge caching~\cite{8374917}. \par
\begin{figure}[t]
	\centerline{\includegraphics[width=0.4\textwidth]{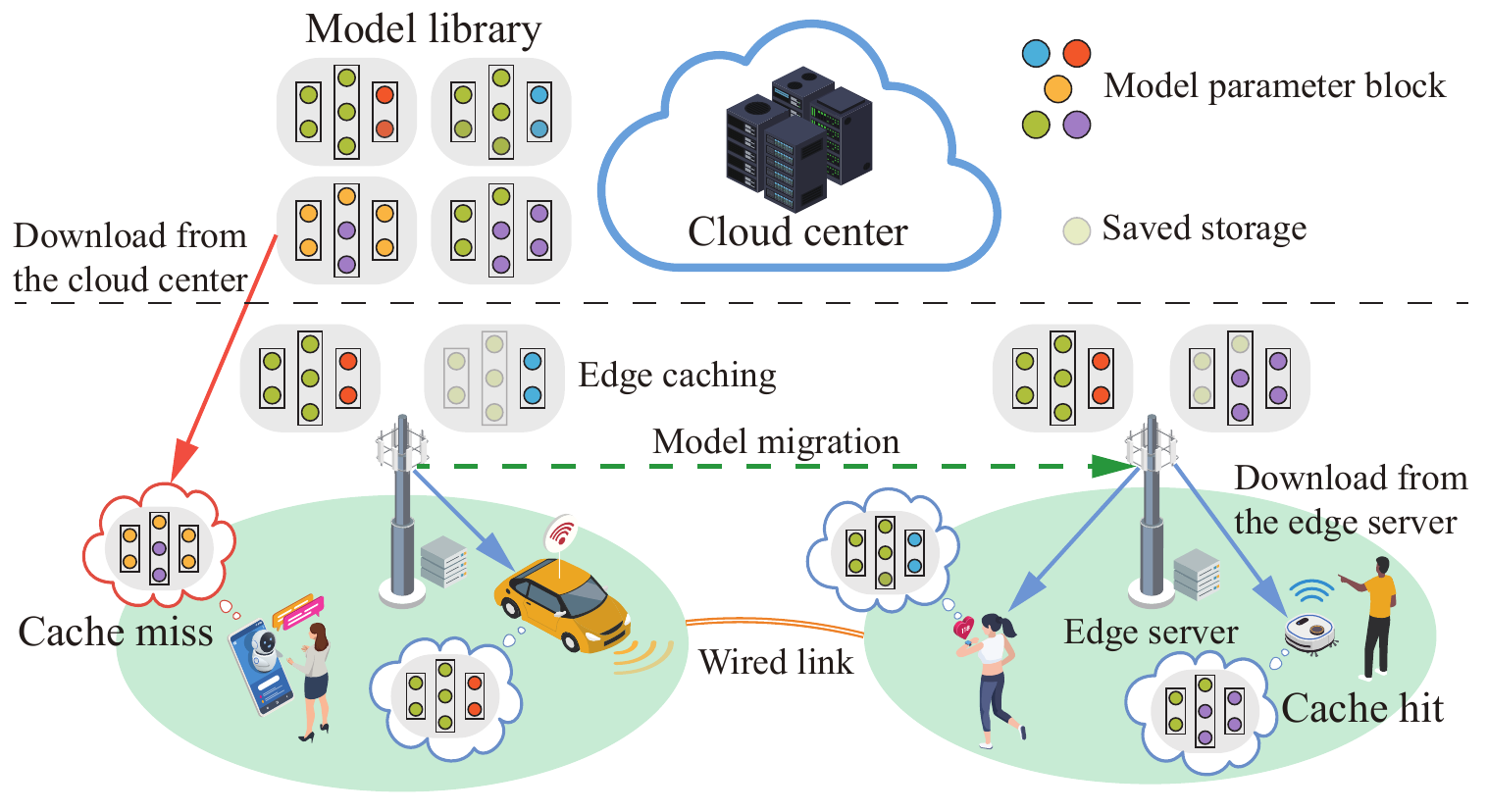}}
	\caption{The TrimCaching framework in multi-edge scenario.}
	\label{fig_system_model}
\end{figure}

\section{Cache Hit Ratio Maximization Problem}
\subsection{Problem Formulation}
The TrimCaching framework aims to maximize the expected cache hit ratio for users' model requests by addressing the model placement problem under the storage capacity of edge servers and service latency requirements. The expected cache hit ratio is 
\begin{equation}\label{original_objective}
U\left({\bf{X}}\right)=\frac{\sum\limits_{k\in\mathcal{K}}\sum\limits_{i\in\mathcal{I}}{p}_{k,i}\left[1-\prod\limits_{m\in\mathcal{M}}\left(1-x_{m,i}{\mathbb{I}}_{1}\left(m,k,i\right)\right)\right]}{\sum\limits_{k\in\mathcal{K}}\sum\limits_{i\in\mathcal{I}}{p}_{k,i}},
\end{equation}
where ${p}_{k,i}$ is the request probability of user $k$ for model $i$. $x_{m,i}$ is the placement of model $i$ on edge server $m$, where $x_{m,i} =1$ indicates that model $i$ is cached on edge server $m$. $\prod\limits_{m\in\mathcal{M}}\left(1-x_{m,i}\mathbb{I}_{1}\left(m,k,i\right)\right) = 0$ implies user $k$ can obtain model $i$ from any edge server within latency requirement, where  
\begin{equation}
\mathbb{I}_{1}\left(m,k,i\right)={\mathbb{I}}_{\left\{T_{m,k,i}\le {\bar{T}}_{k,i}\right\}}
\end{equation} 
is an indicator function, and ${\mathbb{I}}_{1}\left(m,k,i\right) = 1$ if and only if $T_{m,k,i}\le {\bar{T}}_{k,i}$. $T_{m,k,i}$ is the E2E latency when edge server $m$ delivers model $i$ to user $k$, including model downloading latency and on-device inference latency. If user $k$ can download model $i$ from its associated edge server, we have
\begin{equation}
T_{m,k,i}=\frac{D_i}{\bar{C}_{m,k}} + t_{k,i},
\end{equation}
where $D_i$ is the size of model $i$ and $t_{k,i}$ is the inference latency of user $k$ with model $i$. If edge server $m$, which is unassociated with user $k$, supplies model $i$ for user $k$, we have
\begin{equation}
T_{m,k,i}=\mathop{\min}\limits_{m'\in\mathcal{M}_k}\left(\frac{D_i}{C_{m,m'}}+\frac{D_i}{\bar{C}_{m',k}}\right) +t_{k,i}.
\end{equation}

The problem formulation is given as follows. 
\begin{subequations}
	\begin{equation}
		{\mathcal{P}1.1}:\ \mathop{\max}\limits_{{\bf{X}}}\ U\left({\bf{X}}\right)
	\end{equation}	
	\begin{equation}\label{general_problem_storage}
		{\rm{s.t.}} \sum\limits_{j\in\mathcal{J}} D'_j\left[1-\prod\limits_{i\in\mathcal{I}_j}\left(1-x_{m,i}\right) \right]\le Q_m,\ \forall m\in{\mathcal{M}},
	\end{equation}	
	\begin{equation}\label{general_problem_binary}
		x_{m,i}\in\left\{0,1\right\},\ \forall m\in{\mathcal{M}},\forall i\in{\mathcal{I}},
	\end{equation}	
\end{subequations}
where $D'_j$ is the size of parameter block $j$, and the storage capacity of edge server $m$ is denoted by $Q_m$. $1-\prod\limits_{i\in\mathcal{I}_j}\left(1-x_{m,i}\right) = 1$ means that parameter block $j$, if shared by multiple models on edge server $m$, is stored only once, which leads to improved storage efficiency.

Note that our formulation ignores user mobility because the problem is solved based on a ``snapshot'' of user locations. This is commonly adopted in model placement schemes~\cite{6600983}. In practice, our algorithm can conduct model placement decisions by solving the above problem and then re-initiate model placement when the performance degrades to a certain threshold. Our simulation results will show that our algorithm is resilient to user mobility over time, thus eliminating the need for frequent model replacement that would consume backbone bandwidth.

\subsection{Problem Mapping}
Solving $\mathcal{P}1.1$ is extremely challenging due to the product of integer decision variables arising from parameter block sharing. In this subsection, we map the problem to a known NP-hard problem and show there is no polynomial-time algorithm to solve it with a constant approximation ratio. \par 
\begin{proposition}\label{proposition_submodular}
	${\mathcal{P}1.1}$ is a submodular maximization problem with $M$ submodular constraints.
\end{proposition}
\begin{proof}
We first introduce the definition of submodularity and supermodularity before presenting the proof. A set function $f: 2^{X}\rightarrow\mathbb{R}$ is called submodular \cite{fujishige2005submodular} if $f\left(S\right)+f\left(T\right)\ge f\left(S\cup T\right) +f\left(S\cap T\right)$ holds for all subsets $S,T\subseteq X$, where $X$ is a finite set. Additionally, $f$ is submodular if and only if $f\left(S\cup\left\{x\right\}\right)-f\left(S\right)\ge f\left(T\cup\left\{x\right\}\right)-f\left(T\right)$ for all $S\subseteq T$ and $x\in X\setminus T$. The function $f\left(x\right)$ is called supermodular \cite{lovasz1983submodular} if the reversed inequalities hold for every pair of subsets. \par 
	On the one hand, the objective function of $\mathcal{P}1.1$ is a submodular function. For any subset ${\bf{X}}$ of the universal set, we use $\mathcal{I}_{{\bf{X}}}=\left\{i\mid x_{m,i}\in{\bf{X}}\right\}$ and $\mathcal{M}_{{\bf{X}}}=\left\{m\mid x_{m,i}\in{\bf{X}}\right\}$ to represent the placed model set and the edge server set for model placement decision ${\bf{X}}$. Suppose that ${\bf{X}}\subseteq{\bf{X}}'\subseteq{\bf{V}}$ in $\mathcal{P}1.1$, where ${\bf{V}}$ is the universal set that contains all possible model placement results. For any $x_{m,i}\in{\bf{V}}\setminus{\bf{X}}'$, the increased cache hit ratio when we add $x_{m,i}$ into ${\bf{X}}'$ and ${\bf{X}}$ are $U\left({\bf{X}}'\cup\left\{x_{m,i}\right\}\right) - U\left({\bf{X}}'\right)$ and $U\left({\bf{X}}\cup\left\{x_{m,i}\right\}\right) - U\left({\bf{X}}\right)$, respectively. $U\left({\bf{X}}\cup\left\{x_{m,i}\right\}\right) - U\left({\bf{X}}\right)$ is at least larger than $U\left({\bf{X}}'\cup\left\{x_{m,i}\right\}\right) - U\left({\bf{X}}'\right)$. It is because ${\bf{X}}\subseteq{\bf{X}}'$, $\mathcal{M}_{{\bf{X}}}\subseteq\mathcal{M}_{{\bf{X}}'}$, and $\mathcal{I}_{{\bf{X}}}\subseteq\mathcal{I}_{{\bf{X}}'}$. $\mathcal{M}_{{\bf{X}}'}$ can provide at least as many models as $\mathcal{M}_{{\bf{X}}}$ to users. Users can download model $i$ from $\mathcal{I}_{{\bf{X}}'}\setminus\mathcal{I}_{{\bf{X}}}$ provided by $\mathcal{M}_{{\bf{X}}'}$. Therefore, we can conclude that the objective function of $\mathcal{P}1.1$ is a submodular function. \par 
	On the other hand, the constraint function is a set of submodular functions. Suppose that ${\bf{X}}^1_{m}\subseteq{\bf{X}}^2_{m}\subseteq{\bf{V}}_{m}$, and  
    \begin{equation}
    g_m\left({\bf{X}}\right)= \sum\limits_{j\in\mathcal{J}} D'_j\left[1-\prod\limits_{i\in\mathcal{I}_j}\left(1-x_{m,i}\right) \right], 
    \end{equation}
    where ${\bf{V}}_{m}$ is the universal set for edge server $m$. For any element $x_{m,i}\in{\bf{V}}_{m}\setminus{\bf{X}}^2_{m}$, the incremental occupied storage is $g_m\left({\bf{X}}^1_{m}\cup\left\{x_{m,i}\right\}\right)-g_m\left({\bf{X}}^1_{m}\right)$ when we add $x_{m,i}$ into ${\bf{X}}^1_{m}$. Since ${\bf{X}}^1_{m}\subseteq{\bf{X}}^2_{m}$, it is likely that the parameter blocks in model $i$ have already been included in models corresponding to ${\bf{X}}^2_{m}$. Therefore, $g_m\left({\bf{X}}^1_{m}\cup\left\{x_{m,i}\right\}\right)-g_m\left({\bf{X}}^1_{m}\right)\ge g_m\left({\bf{X}}^2_{m}\cup\left\{x_{m,i}\right\}\right)-g_m\left({\bf{X}}^2_{m}\right)$, implying that the constraint function of $\mathcal{P}1.1$ is a set of submodular functions.
\end{proof}
\begin{proposition}\label{NPhard}
	${\mathcal{P}1.1}$ is an NP-hard problem. Moreover, no polynomial-time algorithm that solves ${\mathcal{P}1.1}$ with a constant-approximation guarantee exists.
\end{proposition}
\begin{proof}
We first introduce two concepts. First, $\sum\limits_{x\in X}g\left(x\right)\le b$ is called a knapsack constraint, where $g\left(x\right)$ is a cost function and $X$ is a finite set. Second, if the output of an algorithm for a maximization problem is guaranteed to be at least $\alpha$-times of the optimal solution, where $\alpha\in\left[0,1\right]$ is a constant, we can say this algorithm ensures a constant approximation guarantee. \par 
	$\mathcal{P}1.1$ is equal to the following problem $\mathcal{P}1.2$, where $y_{m,j}$ represents the caching status of parameter block $j$ on edge server $m$. $y_{m,j}=1$ indicates that edge server $m$ stores parameter block $j$. The relationship between $x_{m,i}$ and $y_{m,j}$ are: $x_{m,i}=\prod\limits_{j\in\mathcal{J}_i}y_{m,j}$ and $y_{m,j} = 1-\prod\limits_{i\in\mathcal{I}_j}\left(1-x_{m,i}\right)$. With this relationship, $U\left({\bf{X}}\right)$ is transformed into $U\left(\bf{Y}\right)=\frac{\sum\limits_{k\in\mathcal{K}}\sum\limits_{i\in\mathcal{I}}p_{k,i}\left[1-\prod\limits_{m\in\mathcal{M}}\left(1-\prod\limits_{j\in\mathcal{J}_i}y_{m,j}{\mathbb{I}}_{1}\left(m,k,i\right)\right)\right]}{\sum\limits_{k\in\mathcal{K}}\sum\limits_{i\in\mathcal{I}}{p}_{k,i}}$.
	\begin{subequations}
		\begin{equation}
			{\mathcal{P}1.2}:\ \mathop{\max}\limits_{\bf{Y}}\ U\left(\bf{Y}\right)
		\end{equation}	
		\begin{equation} 
			{\rm{s.t.}} \ \sum\limits_{j\in\mathcal{J}} D'_jy_{m,j}\le Q_m,\ \forall m\in{\mathcal{M}},
		\end{equation}	
		\begin{equation} 
			y_{m,j}\in\left\{0,1\right\},\ \forall m\in{\mathcal{M}},\forall j\in{\mathcal{J}}.
		\end{equation}	
	\end{subequations}\par
	$\mathcal{P}1.2$ is a supermodular maximization problem with $M$ knapsack constraints. The proof is omitted here since it is similar to the proof of Proposition \ref{proposition_submodular}. It has been shown that the supermodular maximization problem with $M$ knapsack constraints is an NP-hard problem, and there is no approximation algorithm with a constant-approximation guarantee to solve this problem when $M=1$ \cite{KELLERER201764}. Therefore, there is no polynomial-time algorithm with a constant-approximation guarantee for $\mathcal{P}1.2$, which also holds for $\mathcal{P}1.1$. 
\end{proof}
Although the problem cannot be solved approximately in general, in the following sections, we will first introduce a special case of $\mathcal{P}1.1$, which represents the typical parameter sharing scenario in practice, for which a polynomial-time algorithm can be developed to obtain a solution with $(1-\epsilon)/2$ approximation guarantee. Subsequently, we will propose a greedy algorithm to solve the original problem $\mathcal{P}1.1$ for the general case. Although there is no constant approximation guarantee, the solution approach is still highly effective.
\section{Special Case with A Small Fixed Number of Shared Parameter Blocks}
We first consider a special case of $\mathcal{P}1.1$ where there is a small fixed number of shared parameter blocks. Formally speaking, the number of shared blocks is independent of the size of the model library and hence the problem scale, making it feasible to traverse all the shared parameter blocks. An example is illustrated in Fig. \ref{fig_tree_model}, where all shared parameter blocks originated from 2 pre-trained models. Such cases are often true in practice because downstream AI models can be derived from a small number of pre-trained models (e.g., Open-source models pre-trained on ImageNet or foundation models, such as GPT-3). Our key idea is to develop an approximate algorithm by traversing the combinations of all shared parameter blocks while judiciously selecting specific parameter blocks, resulting in a polynomial-time algorithm. To achieve this goal, we will propose a successive greedy method to decompose $\mathcal{P}1.1$ into sub-problems, and then solve them with a proposed DP-based algorithm sequentially. The resultant algorithm has $\frac{1-\epsilon}{2}$ approximation guarantee.  
\begin{figure}[t]
	\centerline{\includegraphics[width=0.4\textwidth]{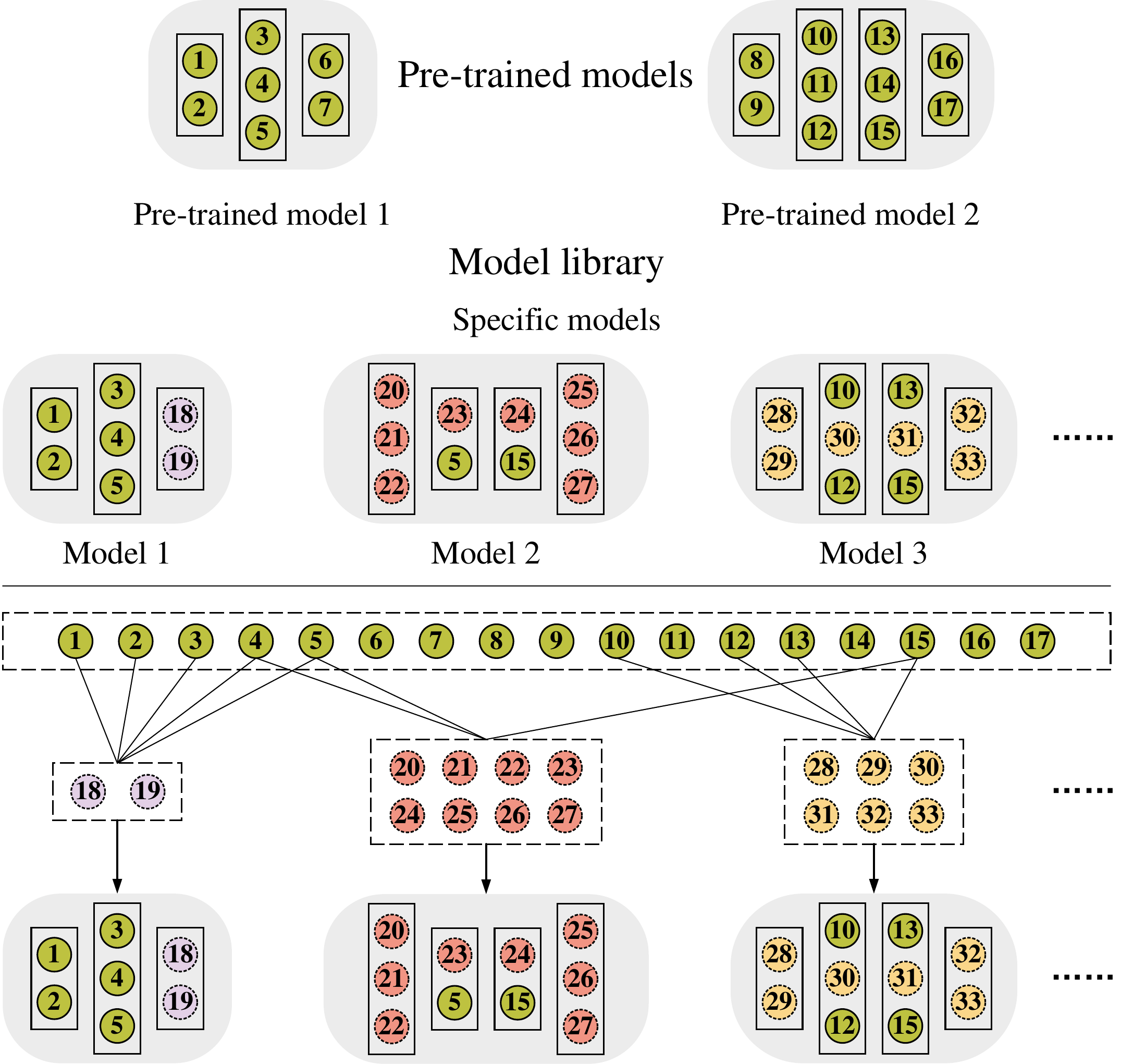}}
	\caption{An example of the special case with a small fixed number of shared parameter blocks. In the figure, regardless of the scale of the model library, the shared green parameter blocks come from two pre-trained models. Nodes in other colors represent specific parameter blocks in the library.
 }
	\label{fig_tree_model}
\end{figure}

\subsection{Successive Greedy Algorithm and the Constant-approximation Guarantee}
We propose a successive greedy algorithm to decompose $\mathcal{P}1.1$ into $M$ sub-problems, each corresponding to edge server $m$, and solve each sub-problem sequentially. The $m$-th sub-problem is
\begin{subequations}
	\begin{equation}
		{\mathcal{P}2.1}_m:\ \mathop{\max}\limits_{\hat{{\bf{X}}}_m}\ \hat{U}_m\left(\hat{{\bf{X}}}_m\right)
	\end{equation}	
	\begin{equation}
		{\rm{s.t.}} \ \sum\limits_{j\in\mathcal{J}} D'_j\left[1-\prod\limits_{i\in\mathcal{I}_j}\left(1-\hat{x}_{m,i}\right) \right]\le Q_m,
	\end{equation}	
	\begin{equation} 
		\hat{x}_{m,i}\in\left\{0,1\right\},\ \forall i\in\mathcal{I},
	\end{equation}	
\end{subequations}
where 
\begin{equation}\label{eq_uhat}
	\begin{aligned}
	\hat{U}_m\left(\hat{{\bf{X}}}_m\right)
	=\frac{\sum\limits_{k\in\mathcal{K}}\sum\limits_{i\in\mathcal{I}}{p}_{k,i}\hat{x}_{m,i}\mathbb{I}_{1}\left(m,k,i\right)\mathbb{I}_{2}\left(m,k,i\right)}{\sum\limits_{k\in\mathcal{K}}\sum\limits_{i\in\mathcal{I}}{p}_{k,i}},
	\end{aligned}
\end{equation}
and 
\begin{equation}
\mathbb{I}_{2}\left(m,k,i\right)=\mathbb{I}_{\left\{\prod\limits_{m'=1}^{m-1}\left(1-\hat{x}_{m',i} \mathbb{I}_{\left\{T_{m',k,i}\le\bar{T}_{k,i}\right\}}\right)=1\right\}}
\end{equation}
is an indicator function, where $\mathbb{I}_{2}\left(m,k,i\right)=1$ represents that the model request for model $i$ of user $k$ cannot be satisfied by any of the first $m-1$ edge servers. Moreover, we assume that $\mathbb{I}_{2}\left(m,k,i\right) = 1$ when $m=1$.

${\mathcal{P}2.1}_m$ aims to design caching decision $\hat{{\bf{X}}}_{m}$ for edge server $m$. $\hat{{\bf{X}}}=\bigcup\limits_{m\in\mathcal{M}}\hat{{\bf{X}}}_m$ is denoted by the solution to ${\mathcal{P}1.1}$ in the special case. Thus, the cache hit ratio of ${\mathcal{P}1.1}$ follows from
\begin{equation}
U\left(\hat{{\bf{X}}}\right) = U\left(\bigcup\limits_{m\in\mathcal{M}}\hat{{\bf{X}}}_m\right) = \sum\limits_{m\in\mathcal{M}}\hat{U}_m\left(\hat{{\bf{X}}}_m\right),
\end{equation}
where the second equality holds because the cache hit ratio for model requests that can be served by multiple edge servers is only counted once, a condition enforced by $\mathbb{I}_{2}\left(m,k,i\right)$. The successive greedy algorithm solves ${\mathcal{P}2.1}_m$ sequentially in the order of $m$, and therefore the placement results of the first $m-1$ edge servers are known to edge server $m$, i.e., $\hat{x}_{m',i}$ is known in  $\mathcal{P}2.1_m$. Edge server $m$ can select and cache a subset of models from model library $\mathcal{I}$ after the first $m-1$ edge servers make caching decisions. The proposed TrimCaching Spec algorithm is summarized in Algorithm \ref{algorithm_successive_greedy}.  \par 
\begin{algorithm}[!ht]
	\caption{TrimCaching Spec} 
	\label{algorithm_successive_greedy}
	\LinesNumbered
	\KwIn{$\mathcal{I}$, $\mathcal{K}$, and $\mathcal{M}$.}
	\KwOut{$\hat{{\bf{X}}}=\bigcup\limits_{m\in\mathcal{M}}\hat{{\bf{X}}}_m$.} 
	{\bf Initialize:} $\hat{{\bf{X}}}_m=\emptyset, \ \forall m\in\mathcal{M}$ and $\mathbb{I}_{2}\left(m,k,i\right)=1$.\\
	\For{$m\in\mathcal{M}$}
	{
		Solve $\mathcal{P}2.1_m$ with Algorithm \ref{algorithm_DP} to obtain $\hat{U}_m$ for edge server $m$, and the corresponding caching decision $\hat{{\bf{X}}}_m$.\\
		Update $\mathbb{I}_{2}\left(m,k,i\right)$.\\
	}
\end{algorithm}\par 
Additionally, we have the following proposition on the constant-approximation guarantee for the TrimCaching Spec algorithm.
\begin{proposition}\label{Successive_greedy_optimal}
	Suppose that each sub-problem can be solved optimally. In the considered special case, the TrimCaching Spec algorithm can obtain a solution $\hat{{\bf{X}}}$ for $\mathcal{P}1.1$ which is lower bounded by $U\left(\hat{{\bf{X}}}\right)\ge\frac{1}{2}U\left({\bf{X}}^*\right)$, where ${\bf{X}}^*$ is the optimal solution to $\mathcal{P}1.1$.
\end{proposition}
\begin{proof}
	The proof is omitted due to page limitations, which can be found in the extended version of this work \cite{qu2024trimcaching}.
\end{proof}
Thus, a condition to achieve the constant-approximation guarantee of the TrimCaching Spec algorithm is obtaining the optimal solution to $\mathcal{P}2.1_m$. 
In the next subsection, we will introduce how to solve $\mathcal{P}2.1_m$ optimally in polynomial time. \par
\subsection{DP-based Approach and the $\epsilon$-optimal Solution}
We propose a DP-based algorithm to solve $\mathcal{P}2.1_m$. The submodularity of the constraint in $\mathcal{P}2.1_m$ results from the shared parameter blocks among models. To address the submodular constraints, our core idea is to make the caching decisions for shared parameter blocks and specific parameter blocks separately. We use $\mathcal{A}$ to denote a set where each element represents a combination of shared parameter blocks. All the possible combinations are included in $\mathcal{A}$. For example, $\{1,2,3,4,5,15\}$ and $\{10,12,13,15\}$ are two elements in $\mathcal{A}$ for Fig. \ref{fig_tree_model}. For a combination $\mathcal{N}$, all models that include $\mathcal{N}$ is denoted by $\mathcal{I}_\mathcal{N}$. 
For example, if $\mathcal{N}$ is $\{1,2,3,4,5,15\}$, $\mathcal{I}_\mathcal{N}$ is the set of model 1 and 2 for Fig. \ref{fig_tree_model}. Additionally, we use $d_{\mathcal{N}}$ to denote the occupied capacity of edge server $m$ for the shared parameter block combination $\mathcal{N}$. In the above example, $d_{\mathcal{N}}$ is the total size of parameter blocks 1, 2, 3, 4, 5, 15. \par 
For ease of calculation, the size of model $i\in\mathcal{I}_{\mathcal{N}}$ that needs to be considered in the further DP process is denoted by 
\begin{equation}\label{eq_size}
D_{\mathcal{N}}\left(i\right)=D_i-d_{\mathcal{N},i},
\end{equation}
where $D_i$ is the size of model $i$, and $d_{\mathcal{N},i}$ is the size of shared parameter blocks in model $i$. For example, for Fig. \ref{fig_tree_model}, if $\mathcal{N}$ is $\{1,2,3,4,5,15\}$, $d_{\mathcal{N},1}$ and $D_{\mathcal{N}}\left(1\right)$ represent the sizes of parameter blocks $\{1,2,3,4,5\}$ and $\{18,19\}$, respectively. \par 
For each $\mathcal{N}$, the proposed DP algorithm determines the maximum cache hit ratio of models in $\mathcal{I}_\mathcal{N}$ within the remaining storage capacity of edge server $m$, i.e. $Q_m-d_{\mathcal{N}}$, and the caching decision is the corresponding set of selected models. Since the denominator in $\hat{U}_m\left(\hat{{\bf{X}}}_m\right)$ is a constant, maximizing $\hat{U}_m\left(\hat{{\bf{X}}}_m\right)$ is equivalent to maximizing the expected number of cache hits $\sum\limits_{k\in\mathcal{K}}\sum\limits_{i\in\mathcal{I}}{p}_{k,i}\hat{x}_{m,i}\mathbb{I}_{1}\left(m,k,i\right)\mathbb{I}_{2}\left(m,k,i\right)$. The expected number of cache hits when placing model $i$ on edge server $m$ is 
\begin{equation}\label{eq_utility}
	u\left(m,i\right)=\sum\limits_{k\in\mathcal{K}}{p}_{k,i}\mathbb{I}_{1}\left(m,k,i\right)\mathbb{I}_{2}\left(m,k,i\right).
\end{equation}
Assume that $u\left(m,i\right)$ is a fixed-point value, and the granularity of $u\left(m,i\right)$ for models in $\mathcal{I}_{\mathcal{N}}$ is $\delta_{m,\mathcal{N}}$, which is related to the precision of $u\left(m,i\right)$. In the above example, $\mathcal{I}_{\mathcal{N}}=\{1,2\}$. If $u\left(m,1\right) = 0.12$ and $u\left(m,2\right)=0.14$, the precision of $u\left(m,1\right)$ and $u\left(m,2\right)$ is two decimal places, and $\delta_{m,\mathcal{N}}=0.01$. Thus, we can say that the number of cache hits of edge server $m$ has at most $\frac{\sum\limits_{i\in\mathcal{I}_\mathcal{N}}u\left(m,i\right)}{\delta_{m,\mathcal{N}}} + 1$ possible values for the currently traversed $\mathcal{N}$. Let $\mathcal{T}\left(e_{\mathcal{N}},w_{\mathcal{N}}\right)$ be the smallest data size of specific parameter blocks that have to be cached on edge server $m$ for achieving the $w_{\mathcal{N}}$-th value of the number of cache hits with the first $e_{\mathcal{N}}$ models in $\mathcal{I}_\mathcal{N}$, where $e_{\mathcal{N}}\in\left\{0,1,\dots,\left|\mathcal{I}_\mathcal{N}\right|\right\}$, and $w_{\mathcal{N}}\in\left\{0,1,\dots,\frac{\sum\limits_{i\in\mathcal{I}_\mathcal{N}}u\left(m,i\right)}{\delta_{m,\mathcal{N}}}\right\}$. Moreover, the initial value of $\mathcal{T}\left(e_{\mathcal{N}},w_{\mathcal{N}}\right)$ is
\begin{equation}
\mathcal{T}\left(e_{\mathcal{N}},w_{\mathcal{N}}\right) = \left\{ {\begin{array}{*{20}{c}}
		{\infty,\ \text{if } w_{\mathcal{N}}\ne0,}\\
		{0,\ \text{if }w_{\mathcal{N}}=0.}
\end{array}} \right.
\end{equation}
The state-transition equation is
\begin{equation}\label{eq_dp}
  \resizebox{1\hsize}{!}{$  \begin{aligned}	&\mathcal{T}\left(e_{\mathcal{N}},w_{\mathcal{N}}\right)\\
    &=\left\{ {\begin{array}{*{20}{c}}
                    \begin{aligned}
				&{\min\left\{
					\begin{aligned}
					&\mathcal{T}\left(e_{\mathcal{N}}-1,w_{\mathcal{N}}\right),\\
					&\mathcal{T}\left(e_{\mathcal{N}}-1,w_{\mathcal{N}}-\frac{u\left(m,e_{\mathcal{N}}\right)}{\delta_{m,\mathcal{N}}}\right)\\
                    &+D_{\mathcal{N}}\left(e_{\mathcal{N}}\right)
					\end{aligned}
					\right\},\ \text{if } w_{\mathcal{N}}\ge\frac{u\left(m,e_{\mathcal{N}}\right)}{\delta_{m,\mathcal{N}}},}\\
				&{\mathcal{T}\left(e_{\mathcal{N}}-1,w_{\mathcal{N}}\right),\ \text{if }w_{\mathcal{N}}<\frac{u\left(m,e_{\mathcal{N}}\right)}{\delta_{m,\mathcal{N}}}.}
                    \end{aligned}
		\end{array}} \right.
    \end{aligned}$}
\end{equation}
Finally, after updating all $\mathcal{T}\left(e_{\mathcal{N}},w_{\mathcal{N}}\right)$, the maximum expected number of cache hits of edge server $m$ with $\mathcal{N}$, i.e., $w^*_{\mathcal{N}}\delta_{m,\mathcal{N}}$, can be determined, where 
\begin{equation}
w^*_{\mathcal{N}}=\mathop{\arg\max}\limits_{w_{\mathcal{N}}}\{w_{\mathcal{N}}\mid\mathcal{T}\left(\left|\mathcal{I}_{\mathcal{N}}\right|,w_{\mathcal{N}}\right) \le Q_m-d_{\mathcal{N}}\}.
\end{equation}

Furthermore, after traversing all feasible combinations of shared parameter blocks in $\mathcal{A}$, we can obtain the maximum expected number of cache hits of edge server $m$, denoted by $w^*_{\mathcal{N}^*}\delta_{m,\mathcal{N}^*}$, where $\mathcal{N}^*=\mathop{\arg\max}\limits_{\mathcal{N}}\{w^*_{\mathcal{N}}\delta_{m,\mathcal{N}}\}$. We denote the model caching decision corresponding to $w^*_{\mathcal{N}^*}\delta_{m,\mathcal{N}^*}$ by $\hat{{\bf{X}}}_m$, and the detailed process of obtaining $\hat{{\bf{X}}}_m$ is shown in the extended version of this work \cite{qu2024trimcaching}. Additionally, the maximum cache hit ratio of edge server $m$ is denoted by
\begin{equation}\label{eq_sp_optimal}
\hat{U}_m\left(\hat{{\bf{X}}}_m\right) = \frac{w^*_{\mathcal{N}^*}\delta_{m,\mathcal{N}^*}}{\sum\limits_{k\in\mathcal{K}}\sum\limits_{i\in\mathcal{I}}{p}_{k,i}}=\frac{\sum\limits_{i\in\hat{\mathcal{I}}_{m}}{u}\left(m,i\right)}{\sum\limits_{k\in\mathcal{K}}\sum\limits_{i\in\mathcal{I}}{p}_{k,i}},
\end{equation} 
where $\hat{\mathcal{I}}_{m} = \left\{i\mid\hat{x}_{m,i}\in\hat{{\bf{X}}}_{m}\right\}$. Thus, the optimal solution to $\mathcal{P}2.1_m$ is obtained.

To further accelerate the computation of \eqref{eq_sp_optimal}, we propose a rounding method to balance precision and computational efficiency. When determining the maximum expected number of cache hits of $\mathcal{I}_\mathcal{N}$, a smaller value of $\delta_{m,\mathcal{N}}$ leads to 
a larger number of feasible $w_{\mathcal{N}}$, further increasing the complexity. To facilitate the execution, we 
round $u\left(m,i\right)$ into 
\begin{equation}
\dot{u}\left(m,i\right) = \lfloor\frac{u\left(m,i\right)}{\epsilon u_{m,\min}}\rfloor,
\end{equation}
where $u_{m,\min} = \mathop{\min}\limits_{i\in\mathcal{I}}u\left(m,i\right)$, and $\epsilon\in\left(0,1\right]$ is a constant. As a result, when determining the maximum expected cache hit ratio for edge server $m$, some $u\left(m,i\right)$ of models in $\mathcal{I}_{\mathcal{N}}$ can be rounded into the same value, and the number of feasible values of $w_{\mathcal{N}}$ can be decreased, thereby accelerating the DP process. We denote $\dot{\delta}_{m,\mathcal{N}}$ by the granularity of $\dot{u}\left(m,i\right)$ for models in $\mathcal{I}_{\mathcal{N}}$. Similarly, $\mathcal{T}\left(e_{\mathcal{N}},\dot{w}_{\mathcal{N}}\right)$ is the smallest data size of specific parameter blocks that have to be cached on edge server $m$ for achieving the $\dot{w}_{\mathcal{N}}$-th value of the cache hits with the first $e_{\mathcal{N}}$ models.
Besides, we define $\dot{w}^*_{\mathcal{N}}=\mathop{\arg\max}\limits_{\dot{w}_{\mathcal{N}}}\{\dot{w}_{\mathcal{N}}\mid\mathcal{T}\left(\left|\mathcal{I}_{\mathcal{N}}\right|,\dot{w}_{\mathcal{N}}\right) \le Q_m-d_{\mathcal{N}}\}$ and $\dot{w}^*_{\mathcal{N}^*}\delta_{m,\mathcal{N}^*}=\mathop{\max}\limits_{\mathcal{N}}\{\dot{w}^*_{\mathcal{N}}\delta_{m,\mathcal{N}}\}$. At last, in $\mathcal{P}2.1_{m}$, the caching decision corresponding to $\dot{w}^*_{\mathcal{N}^*}$ is denoted by $\dot{{\bf{X}}}_m$, and 
\begin{equation}
	\hat{U}_m\left(\dot{{\bf{X}}}_m\right) = \frac{\sum\limits_{i\in\dot{\mathcal{I}}_{m}}{u}\left(m,i\right)}{\sum\limits_{k\in\mathcal{K}}\sum\limits_{i\in\mathcal{I}}{p}_{k,i}},
\end{equation} 
where $\dot{\mathcal{I}}_{m} = \left\{i\mid\hat{x}_{m,i}\in\dot{{\bf{X}}}_{m}\right\}$. The DP-based rounding algorithm is shown in Algorithm \ref{algorithm_DP}. 

\begin{algorithm}[!ht]
	\caption{DP-based Rounding Algorithm} 
	\label{algorithm_DP}
	\LinesNumbered
	\KwIn{$\epsilon$, $\mathcal{K}$, $\mathcal{I}$.}
	\KwOut{$\hat{U}_{m}\left(\dot{{\bf{X}}}_m\right)$, $\dot{{\bf{X}}}_m$.} 
	{\bf Initialize:} $\hat{U}_{m}\left(\dot{{\bf{X}}}_m\right) = 0$, $\dot{{\bf{X}}}_m=\emptyset$. Calculate $u\left(m,i\right)$ from \eqref{eq_utility}.\\
	\For{$\mathcal{N}\in\mathcal{A}$}
	{
		Calculate the occupied storage $d_{\mathcal{N}}$ for $\mathcal{N}$.\\
		\If{$d_{\mathcal{N}}>Q_m$}
		{
			\textbf{Continue}.\\
		}
            \eIf{$\epsilon = 0$}
	      {
		      $\dot{u}\left(m,i\right)=u\left(m,i\right)$.\\
	    }
	    {
		      $\dot{u}\left(m,i\right)=\lfloor\frac{u\left(m,i\right)}{\epsilon u_{m,\min}}\rfloor$.
	    }
		Calculate the granularity $\dot{\delta}_{m,\mathcal{N}}$ of $\dot{u}\left(m,i\right)$.\\
		\For{$e_{\mathcal{N}}\in\mathcal{I}_\mathcal{N}$}
		{
			Calculate $D_{\mathcal{N}}\left(e_{\mathcal{N}}\right)$ from \eqref{eq_size}. \\
		}
		Set the initial value of $\mathcal{T}\left(e_{\mathcal{N}},\dot{w}_{\mathcal{N}}\right)$, $\mathcal{T}\left(e_{\mathcal{N}},\dot{w}_{\mathcal{N}}\right) = \left\{ {\begin{array}{*{20}{c}}
				{\infty,\ \text{if } \dot{w}_{\mathcal{N}}\ne0,}\\
				{0,\ \text{if }\dot{w}_{\mathcal{N}}=0.}
		\end{array}} \right.$\\
		\For{$e_{\mathcal{N}}\in\left\{1,\dots,\left|\mathcal{I}_{\mathcal{N}}\right|\right\}$}
		{
			\For{$\dot{w}_{\mathcal{N}}\in\left\{1,\dots,\frac{\sum\limits_{i\in\mathcal{I}_\mathcal{N}}\dot{u}\left(m,i\right)}{\dot{\delta}_{m,\mathcal{N}}}\right\}$}
			{
				\eIf{$\dot{u}\left(m,e_{\mathcal{N}}\right)\le \dot{w}_{\mathcal{N}}\dot{\delta}_{m,\mathcal{N}}$}
				{
                        $\mathcal{T}\left(e_{\mathcal{N}},\dot{w}_{\mathcal{N}}\right) = \min\left\{
					\begin{aligned}
						&\mathcal{T}\left(e_{\mathcal{N}}-1,\dot{w}_{\mathcal{N}}\right),\\
                            &\mathcal{T}\left(e_{\mathcal{N}}-1,\dot{w}_{\mathcal{N}}-\frac{\dot{u}\left(m,e_{\mathcal{N}}\right)}{\dot{\delta}_{m,\mathcal{N}}}\right) \\
                            &+ D_{\mathcal{N}}\left(e_{\mathcal{N}}\right)
					\end{aligned}
                        \right\}
                        $.\\
				}
				{
					$\mathcal{T}\left(e_{\mathcal{N}},\dot{w}_{\mathcal{N}}\right) = \mathcal{T}\left(e_{\mathcal{N}}-1,\dot{w}_{\mathcal{N}}\right)$. 
				}
			}
		}
		$\dot{w}^*_{\mathcal{N}}=\mathop{\arg\max}\limits_{\dot{w}_{\mathcal{N}}}\{\dot{w}_{\mathcal{N}}\mid\mathcal{T}\left(\left|\mathcal{I}_{\mathcal{N}}\right|,\dot{w}_{\mathcal{N}}\right) \le Q_m-d_{\mathcal{N}}\}$, and the corresponding model caching decision to $\dot{w}^*_{\mathcal{N}}$ is $\dot{{\bf{X}}}_{m,\mathcal{N}}$ . \\
        $\hat{U}_m\left(\dot{{\bf{X}}}_{m,\mathcal{N}}\right)=\frac{\sum\limits_{i\in\dot{\mathcal{I}}_{m,\mathcal{N}}}u\left(m,i\right)}{\sum\limits_{k\in\mathcal{K}}\sum\limits_{i\in\mathcal{I}}{p}_{k,i}}$, where $\dot{\mathcal{I}}_{m,\mathcal{N}} = \left\{i\mid\hat{x}_{m,i}\in\dot{{\bf{X}}}_{m,\mathcal{N}}\right\}$.\\
		\If{$\hat{U}_m\left(\dot{{\bf{X}}}_{m,\mathcal{N}}\right)>\hat{U}_{m}\left(\dot{{\bf{X}}}_m\right)$}
		{
			$\hat{U}_{m}\left(\dot{{\bf{X}}}_m\right)=\hat{U}_m\left(\dot{{\bf{X}}}_{m,\mathcal{N}}\right)$, and $\dot{{\bf{X}}}_m=\dot{{\bf{X}}}_{m,\mathcal{N}}$.\\
		}
	}
\end{algorithm}\par
Algorithm \ref{algorithm_DP} ensures an $\epsilon$-optimal solution. The detailed results are summarized as follows. 
\begin{proposition}\label{DP_epsilon}
	The cache hit ratio produced by Algorithm \ref{algorithm_DP} satisfies $\hat{U}_{m}\left(\dot{{\bf{X}}}_m\right)\ge\left(1-\epsilon\right)\hat{U}_{m}\left(\hat{{\bf{X}}}^*_m\right)$, where $\hat{{\bf{X}}}^*_m$ is the optimal solution to $\mathcal{P}2.1_{m}$ under the successive greedy algorithm.
\end{proposition}
 \begin{proof}
 	The proof is omitted due to page limitations, which can be found in the extended version of this work \cite{qu2024trimcaching}.
 \end{proof}\par 
\begin{theorem}
    The TrimCaching Spec algorithm has a polynomial-time computational complexity $O\left(MI\right)$.
\end{theorem}
\begin{proof}
Assuming $\beta$ is the number of shared parameter blocks, Algorithm \ref{algorithm_DP} requires at most $2^{\beta}$ iterations to traverse all combinations of shared parameter blocks. Moreover, the complexity for updating $\mathcal{T}\left(e_{\mathcal{N}},\dot{w}_{\mathcal{N}}\right)$ in Algorithm \ref{algorithm_DP} is $O\left(I\frac{{p}}{\delta_{\min}}\right)$, where ${p} = \sum\limits_{k\in\mathcal{K}}\sum\limits_{i\in\mathcal{I}}{p}_{k,i}$ and $\delta_{\min}$ is the granularity of $p_{k,i}$. Therefore, the overall complexity of Algorithm \ref{algorithm_DP} is $O\left(2^{\beta}I\frac{{p}}{\delta_{\min}}\right)$. Since we only need to check at most $M$ edge servers in the TrimCaching Spec algorithm, the total complexity of the TrimCaching Spec algorithm is $O\left(2^{\beta}\frac{{p}}{\delta_{\min}}MI\right)$. Note that $\beta$ is a constant independent of the problem scale as mentioned for the considered special case (which is usually small in practice). Thus, the complexity of the TrimCaching Spec algorithm is $O\left(MI\right)$. 
\end{proof}
\begin{theorem}\label{theorem_spec_final}
The TrimCaching Spec algorithm can obtain a solution $\hat{{\bf{X}}}$ for $\mathcal{P}1.1$ satisfying $U\left(\hat{{\bf{X}}}\right)\ge\frac{1-\epsilon}{2}U\left({\bf{X}}^*\right)$.
\end{theorem}
\begin{proof}
    First, $U\left(\hat{{\bf{X}}}\right)=U\left(\bigcup\limits_{m\in\mathcal{M}}\dot{{\bf{X}}}_m\right) = \sum\limits_{m\in\mathcal{M}}\hat{U}_m\left(\dot{{\bf{X}}}_m\right)$ holds. Second, Proposition \ref{DP_epsilon} establishes $\sum\limits_{m\in\mathcal{M}}\hat{U}_m\left(\dot{{\bf{X}}}_m\right)\ge\sum\limits_{m\in\mathcal{M}}\left(1-\epsilon\right)\hat{U}_m\left(\hat{{\bf{X}}}^*_m\right)=\left(1-\epsilon\right)U\left(\bigcup\limits_{m\in\mathcal{M}}\hat{{\bf{X}}}_m^*\right)=\left(1-\epsilon\right)U\left(\hat{{\bf{X}}}^*\right)$, where $\hat{{\bf{X}}}^* = \bigcup\limits_{m\in\mathcal{M}}\hat{{\bf{X}}}_m^*$ is the optimal solution to $\mathcal{P}1.1$ under the successive greedy algorithm. Finally, Proposition \ref{Successive_greedy_optimal} confirms $\left(1-\epsilon\right)U\left(\hat{{\bf{X}}}^*\right)\ge\frac{1-\epsilon}{2}U\left({\bf{X}}^*\right)$, which completes the proof. 
\end{proof}

\section{The General Case: Arbitrary Parameter Sharing}
This section examines the general case of $\mathcal{P}1.1$, where models can arbitrarily share parameter blocks, in contrast to the special case with a small fixed number of shared parameter blocks. Formally speaking, the number of shared parameter blocks among models may increase with the scale of the model library, such that searching for all combinations of shared parameter blocks, as required by the TrimCaching Spec algorithm, results in an exponential time complexity in terms of the problem scale, which should be avoided in the general case.

Our proposed greedy algorithm for $\mathcal{P}1.1$ is outlined in Algorithm \ref{algorithm_greedy}. The key part of the TrimCaching Gen algorithm is in line 4. In the $l$-th step, the algorithm first calculates the incremental cache hit ratio for each $x_{m,i}$ based on the previous placement decision ${\bf{X}}^{l-1}$ made in the $l-1$-th step. Subsequently, the TrimCaching Gen algorithm identifies $\{m^*,i^*\}$ which yields the maximum increased cache hit ratio, i.e.,  $U\left({\bf{X}}^{l-1}\cup \left\{x_{m^*,i^*}\right\}\right)-U\left({\bf{X}}^{l-1}\right)$, while ensuring that $g_m\left({\bf{X}}_m^{l-1}\cup\{x_{m^*,i^*}\}\right)$ does not exceed the edge server capacity $Q_m$, where ${\bf{X}}_m^{l-1}$ is the model placement result of edge server $m$ in the $l-1$-th step. By adding $x_{m^*,i^*}$ into ${\bf{X}}^{l-1}$, we have ${\bf{X}}^l={\bf{X}}^{l-1}\cup\left\{x_{m^*,i^*}\right\}$. The TrimCaching Gen algorithm will repeat the above steps until all edge servers cannot cache any models. The proposed algorithm can solve $\mathcal{P}1.1$ with the time complexity of $O\left(MI\right)$. In addition, we have the following result for the TrimCaching Gen algorithm.\par 
\begin{algorithm}[!ht]
	\caption{TrimCaching Gen} 
	\label{algorithm_greedy}
	\LinesNumbered
	\KwIn{$\mathcal{K}$, $\mathcal{I}$, and $\mathcal{M}$.}
	\KwOut{${\bf{X}}$ and $U\left({\bf{X}}\right)$.} 
	{\bf Initialize:}
	$l=0$, ${\bf{X}}^l=\emptyset$, and $U\left({\bf{X}}^l\right)=0$.\\
	\While {There is an edge server that can continue to cache models.}
	{
		$l=l+1$.\\
			$\left\{m^*,i^*\right\} = \mathop{\arg\max}\limits_{m,i}\{U\left({\bf{X}}^{l-1}\cup \left\{x_{m,i}\right\}\right)-U\left({\bf{X}}^{l-1}\right)\mid g_m\left({\bf{X}}_m^{l-1}\cup\{x_{m,i}\}\right)\le Q_m\}$\\
		${\bf{X}}^l={\bf{X}}^{l-1}\cup\left\{x_{m^*,i^*}\right\}$
	}
	Output ${\bf{X}}={\bf{X}}^l$ and $U\left({\bf{X}}\right)$.\\
\end{algorithm}\par 
\begin{theorem}
Given the the solution to the TrimCaching Gen algorithm ${\bf{X}}=\bigcup\limits_{m\in\mathcal{M}}{\bf{X}}_m$, where ${\bf{X}}_m$ is the model placement decision of edge server $m$, the TrimCaching Gen algorithm can obtain a result satisfying $U\left({\bf{X}}\right)\ge\frac{1}{\Gamma}U\left({\bf{X}}^*\right)$. Here, $\Gamma=\max\{\left|{\bf{X}}\right|:g_m\left({\bf{X}}_m\right)\le Q_m,\forall m\in\mathcal{M}\}$.
\end{theorem}
\begin{proof}
The proof is omitted due to page limitations. A similar proof can be found in \cite{iyer2013submodular,iyer2013fast}.
\end{proof}
Since the lower bound decreases as the number of AI models in the model library and the number of edge servers grows, there is no constant approximation guarantee for the algorithm. This coincides with Proposition \ref{NPhard} that no polynomial-time approximation algorithm exists for the general case.\par

\section{Numerical Results}
In this section, we will conduct simulations for the proposed TrimCaching framework.
\subsection{Simulation Setup}
In the simulation, $K$ users and $M$ edge servers are uniformly distributed in a square area of 1 $\text{km}^\text{2}$. The QoS requirements of model downloading latency and on-device inference latency for users are uniformly distributed in the range $\left[0.5,1\right]\ \text{s}$ \cite{3gpp.22.874}. The number of users is $K=\left\{10,20,30,40,50\right\}$. Suppose the coverage radius of edge servers is 275 m. The associated user set of edge server $m$ is denoted as $\mathcal{K}_m$. The expected bandwidth and transmit power assigned by edge server $m$ for its associated user $k$ are $\bar{B}_{m,k} = \frac{B}{{p}_\text{A}\left|\mathcal{K}_m\right|}$ and $\bar{P}_{m,k} = \frac{P}{{p}_\text{A}\left|\mathcal{K}_m\right|}$, where $B=400 \ \text{MHz}$ and $P=43\ \text{dBm}$ are the total bandwidth and transmit power of an edge server, respectively. We assume the active probability of a user is ${p}_\text{A}=0.5$. The communication data rate between edge servers is $C_{m,m'}=10\ {\text{Gbps}}$ \cite{gsmabachkaul,8488527}. We set $\gamma_0$ and $\alpha_0$ in \eqref{eq_communication} as 1 and 4, respectively. The number of edge servers is $M=\left\{6,8,10,12,14\right\}$. Besides, each edge server has identical storage capacity, denoted by $Q$, ranging from 0.5 to 1.5 GB. Note that the storage capacity of edge servers can be much larger than 1.5 GB in reality. However, the model library can also be significantly larger than our constructed parameter-sharing model library of 300 models, which will be introduced later. Due to our limited computing resources for model fine-tuning, we proportionally reduce the storage capacity of edge servers and the size of the model library, which will not impact the phenomenon observed in the experiments.\par 

The parameter-sharing model library is constructed based on the ResNet family, i.e., ResNet18, ResNet34, and ResNet50~\cite{he2016deep}, and the image classification dataset CIFAR100 \cite{krizhevsky2009learning}. There are 20 superclasses in the dataset, and each superclass includes 5 classes. For example, the superclass ``fish" includes class ``aquarium fish", ``flatfish", ``ray", ``shark", and ``trout". 
We fine-tune these three pre-trained models into 100 downstream classification models, respectively. Each model corresponds to each class, e.g., a classification model for sharks. 
The request probability of each end user for 300 models in each kind of model library obeys the Zipf distribution \cite{zipf1929relative}. 
For the special case and the general case, the parameter sharing among AI models is constructed as follows. 

\begin{itemize}
    \item \textbf{The special case}: In this case, the number of shared parameter blocks is independent of the scale of the model library because all of them are fine-tuned from three pre-trained models. Without loss of generality, we adopt the classic bottom layer freezing techniques to fine-tune these three models for the 100 downstream classification services. The number of frozen bottom layers (each corresponding to a shared parameter block in our system model) for ResNet18, ResNet34, and ResNet50 falls within the ranges of $\left[29,40\right]$, $\left[49,72\right]$, and $\left[87,106\right]$, respectively.
    
    By default, $\epsilon$ is set to 0.1 in Algorithm \ref{algorithm_DP}. 
     \item  \textbf{The general case}: In this case, we first fine-tune all parameters of the ResNet18, ResNet34, and ResNet50 on CIFAR100 with a few selected superclasses. Then, we fine-tune new models on sub-datasets in CIFAR100 with similar classes based on bottom layer freezing. The fine-tuned models for each class within every superclass in the second round reuse parameter blocks from models fine-tuned for the selected superclass in the first round, creating a large set of shared parameter blocks related to the scale of the model library. The fine-tuning settings are detailed in Table \ref{table_ft}.
\end{itemize}
\begin{table}[th]
	\centering
	\caption{Fine-tuning Settings}
	\label{table_ft}
	\begin{tabular}{|p{3cm}|p{5cm}|}
		\hline
		\textbf{First-round fine-tuning} & \textbf{Second-round fine-tuning}\\ \hline
		fruit and vegetables & flowers, trees\\ \hline
		medium-sized mammals & large carnivores, large omnivores and herbivores, people, reptiles, small mammals\\ \hline
		vehicles 2 & large man-made outdoor things, vehicles 1\\ \hline
	\end{tabular}
\end{table} \par 

For comparisons, three algorithms are considered:
\begin{itemize}
    \item \textbf{The TrimCaching Spec algorithm}: Algorithm \ref{algorithm_successive_greedy}, which is developed for the special case. 
     \item  \textbf{The TrimCaching Gen algorithm}: Algorithm \ref{algorithm_greedy}, which is developed for the general case. It can also be applied to the special case.
    \item \textbf{Independent Caching:} AI models are cached independently without considering their shared parameters, referring to traditional content placement schemes \cite{8374917}. 
\end{itemize}\par 
All simulation results presented below are averaged from 100 network topologies. For each network topology, we further simulate the cache hit ratio using over $10^3$ channel realizations over Rayleigh fading. Note that the placement decisions are obtained based on average channel gains while the cache hit performance is examined over the Rayleigh fading model.

\begin{figure*}[!ht]
	\centering
	\subfigure[Cache hit ratio v.s. $Q$, where $M=10$ and $I=30$. ]{\includegraphics[width=2in]{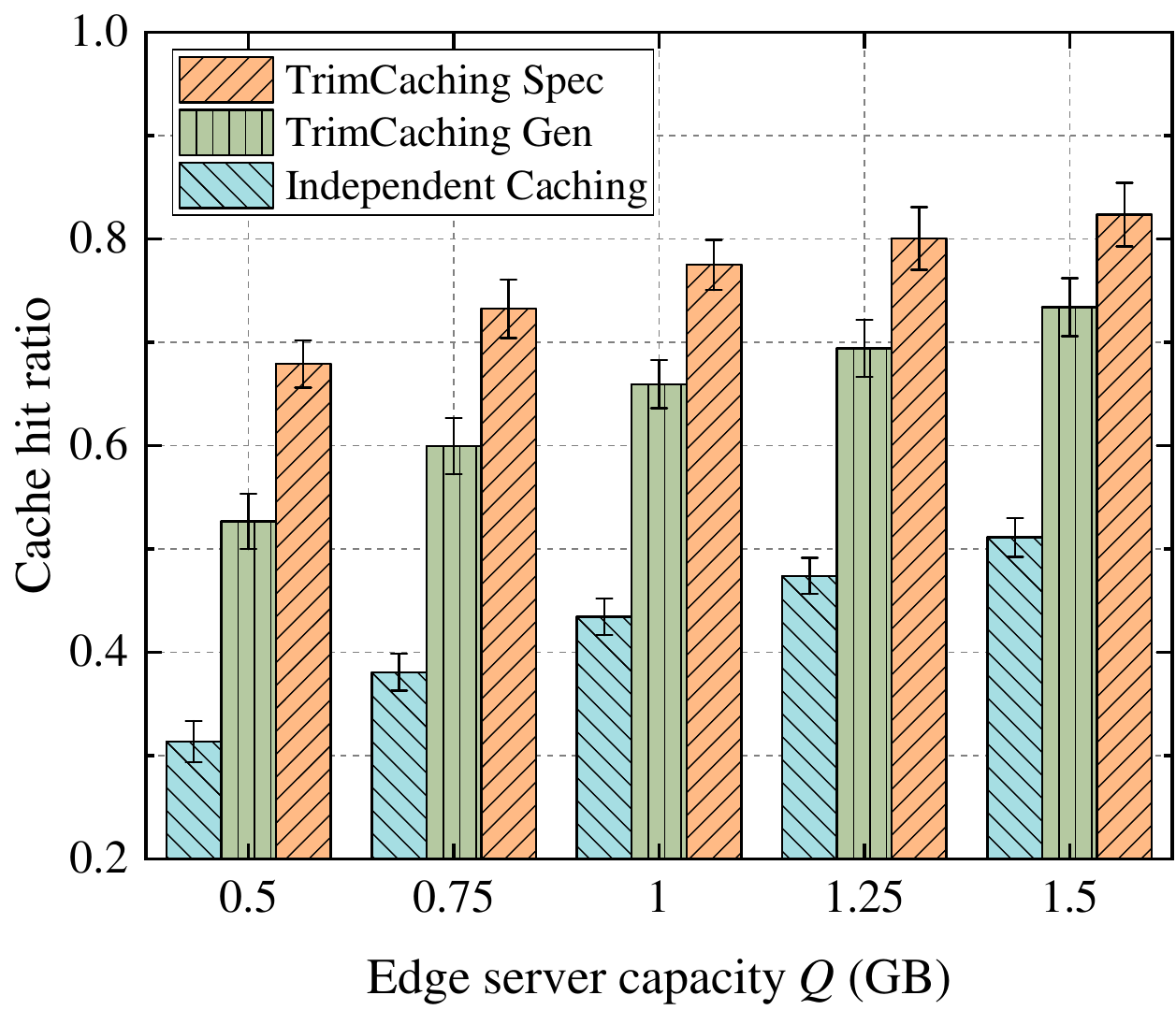}\label{fig_case2_capacity}}
	\quad
	\subfigure[Cache hit ratio v.s. $M$, where $Q=1 \ \text{GB}$ and $I=30$.]{\includegraphics[width=2in]{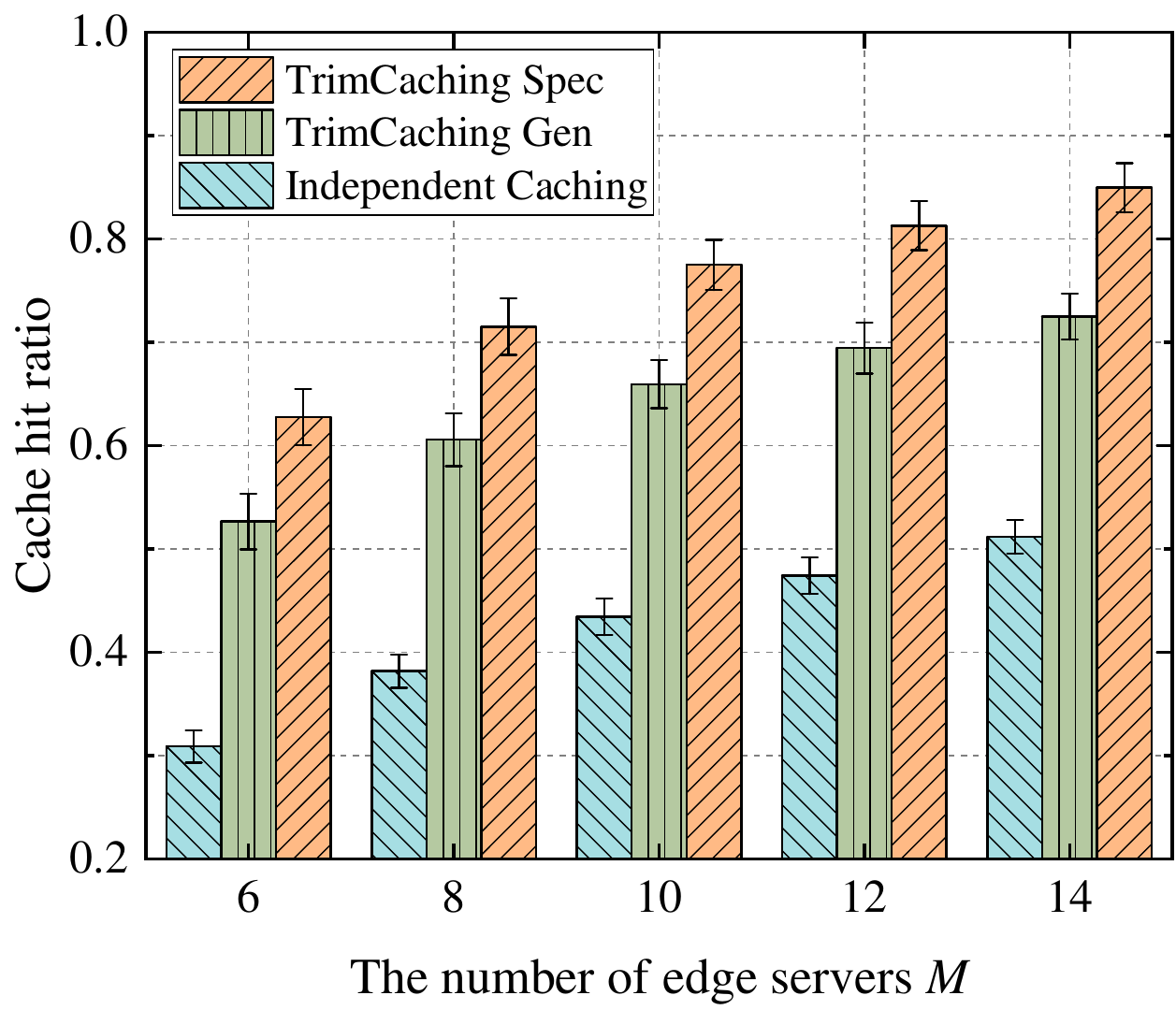}\label{fig_case2_edge}}
	\quad
	\subfigure[Cache hit ratio v.s. $K$, where $Q=1 \ \text{GB}$ and $M=10$.]{\includegraphics[width=2in]{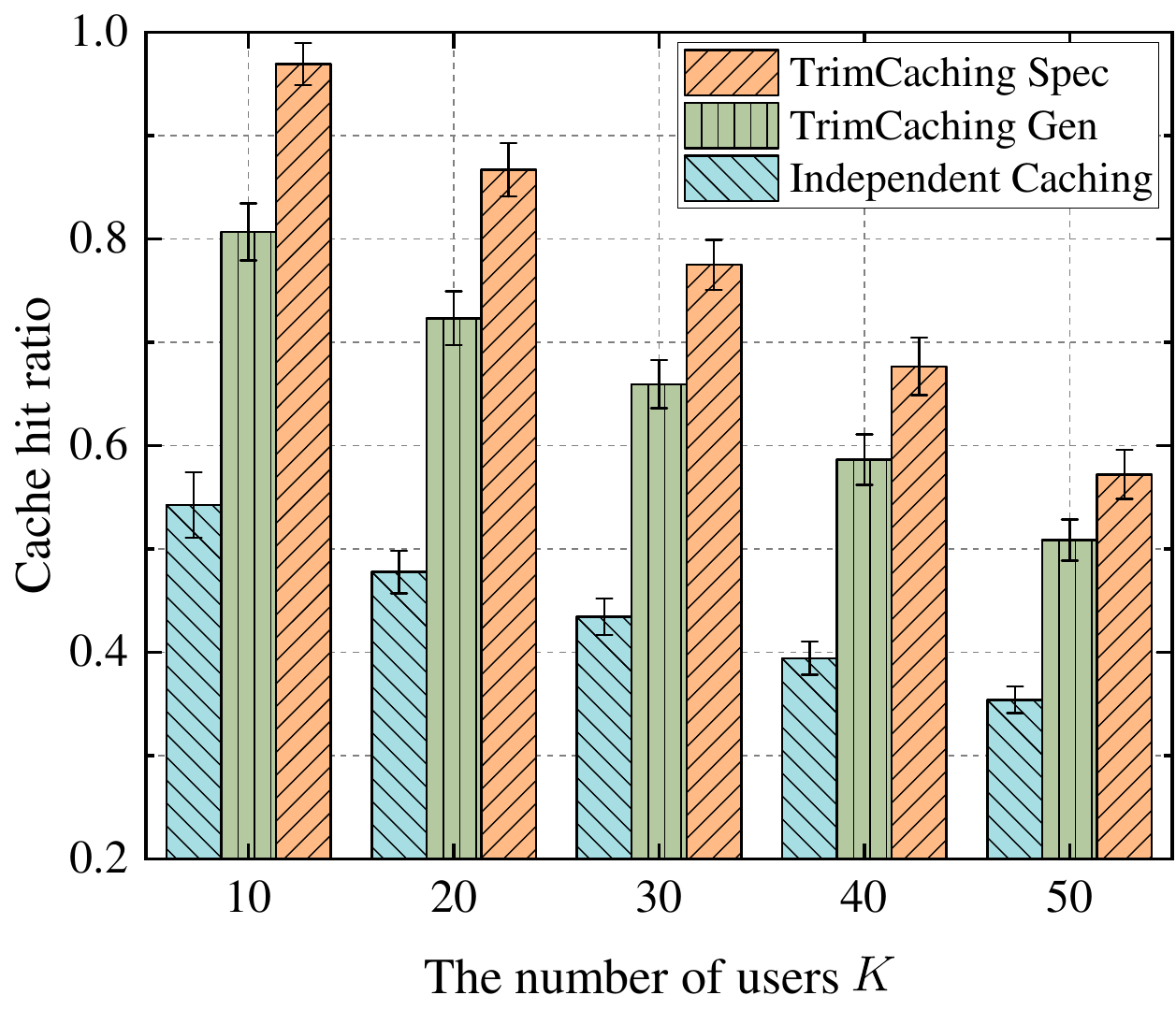}\label{fig_case2_user}}
	\caption{Cache hit ratio for the special case, where a small fixed number of shared parameter blocks is considered. The error bar denotes the standard deviation, which is the same for the subsequent figures.}\label{fig_special_case}
\end{figure*}
\subsection{Performance in the Special Case}
We first compare the algorithms in the special case. It is noted that the TrimCaching Gen algorithm can also be applied to this special case, so we use it as the benchmark.
Fig. \ref{fig_special_case} shows the cache hit ratio of these algorithms. In particular, Fig. \ref{fig_case2_capacity} illustrates the cache hit ratio by varying the storage capacity. As $Q$ increases, the cache hit ratios of all algorithms increase. This is because a larger storage capacity allows edge servers to cache more models to serve end users. Unsurprisingly, both the TrimCaching Spec algorithm and the TrimCaching Gen algorithm perform better than the Independent Caching framework due to the storage efficiency of parameter-sharing caching. Moreover, the proposed TrimCaching Spec algorithm outperforms the TrimCaching Gen algorithm. As shown earlier, the TrimCaching Spec algorithm has a constant-approximation guarantee, while the TrimCaching Gen algorithm does not. The cache hit ratio of the TrimCaching Spec algorithm is 11.93\% and 33.93\% higher than the TrimCaching Gen algorithm and the Independent Caching framework on average. A similar trend can be observed in Fig. \ref{fig_case2_edge}, which varies the number of edge servers.

Fig. \ref{fig_case2_user} investigates the effectiveness of algorithms when the number of users increases. The cache hit ratio decreases as the number of users grows because the transmission data rate decreases. However, the improved cache hit ratio of the proposed algorithms is still significant. Compared with the Independent Caching framework, the increased cache hit ratios of the proposed TrimCaching Spec algorithm and TrimCaching Gen algorithm are about 21.81\% and 15.47\%, respectively, when $K=50$. Also, the average performance gain brought by the TrimCaching Spec algorithm compared with the TrimCaching Gen algorithm is about 11.50\%, indicating that the tailored TrimCaching Spec algorithm achieves better performance in the special case and the TrimCaching Gen algorithm can still achieve competitive performance.  \par 
\subsection{Performance in the General Case}

\begin{figure*}[!ht]
	\centering
	\subfigure[Cache hit ratio v.s. $Q$, where $M=10$ and $I=30$. ]{\includegraphics[width=2in]{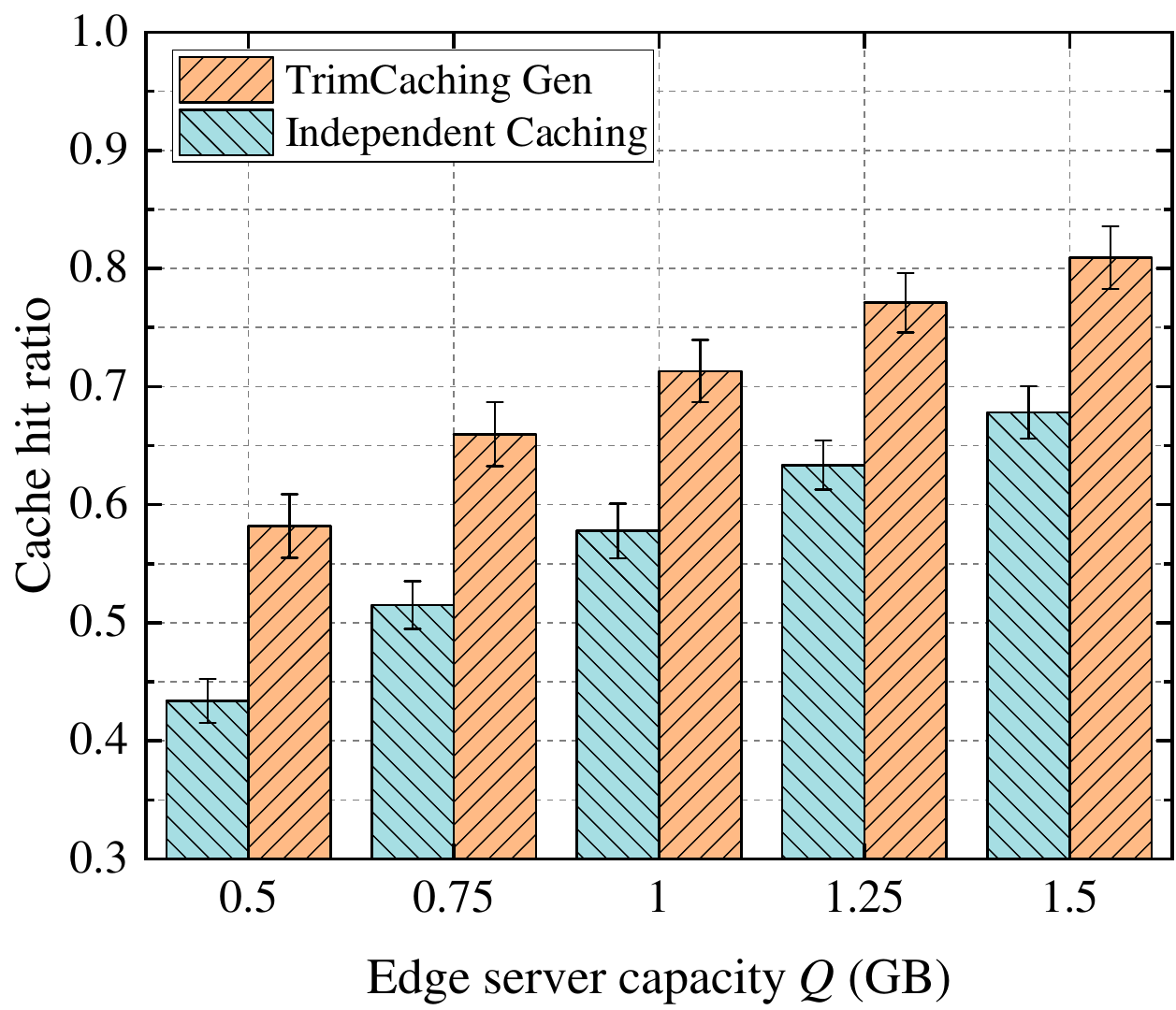}\label{fig_real_case3_capacity}}
	\quad
	\subfigure[Cache hit ratio v.s. $M$, where $Q=1 \ \text{GB}$ and $I=30$.]{\includegraphics[width=2in]{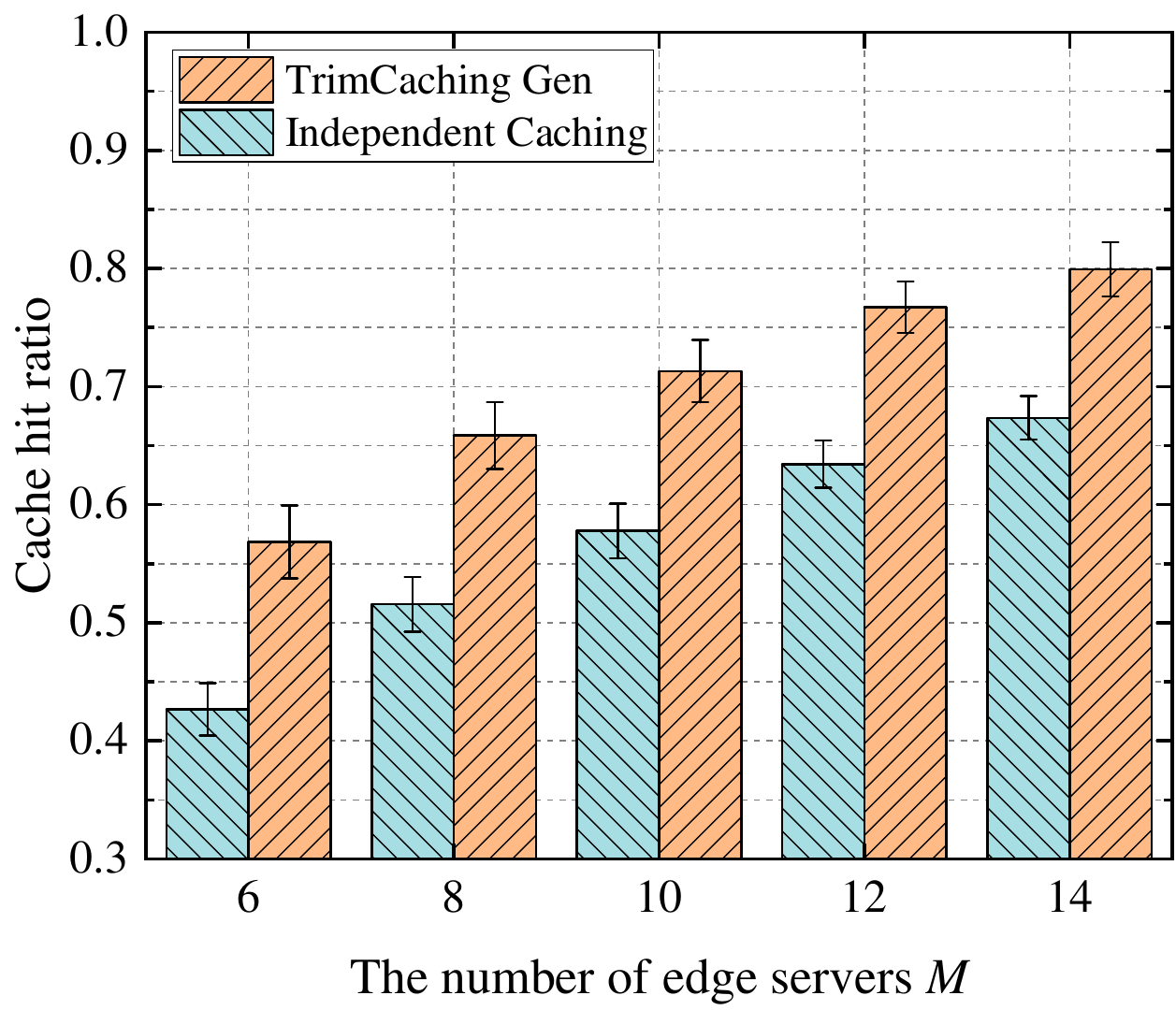}\label{fig_real_case3_edge}}
	\quad
	\subfigure[Cache hit ratio v.s. $K$, where $Q=1 \ \text{GB}$ and $M=10$.]{\includegraphics[width=2in]{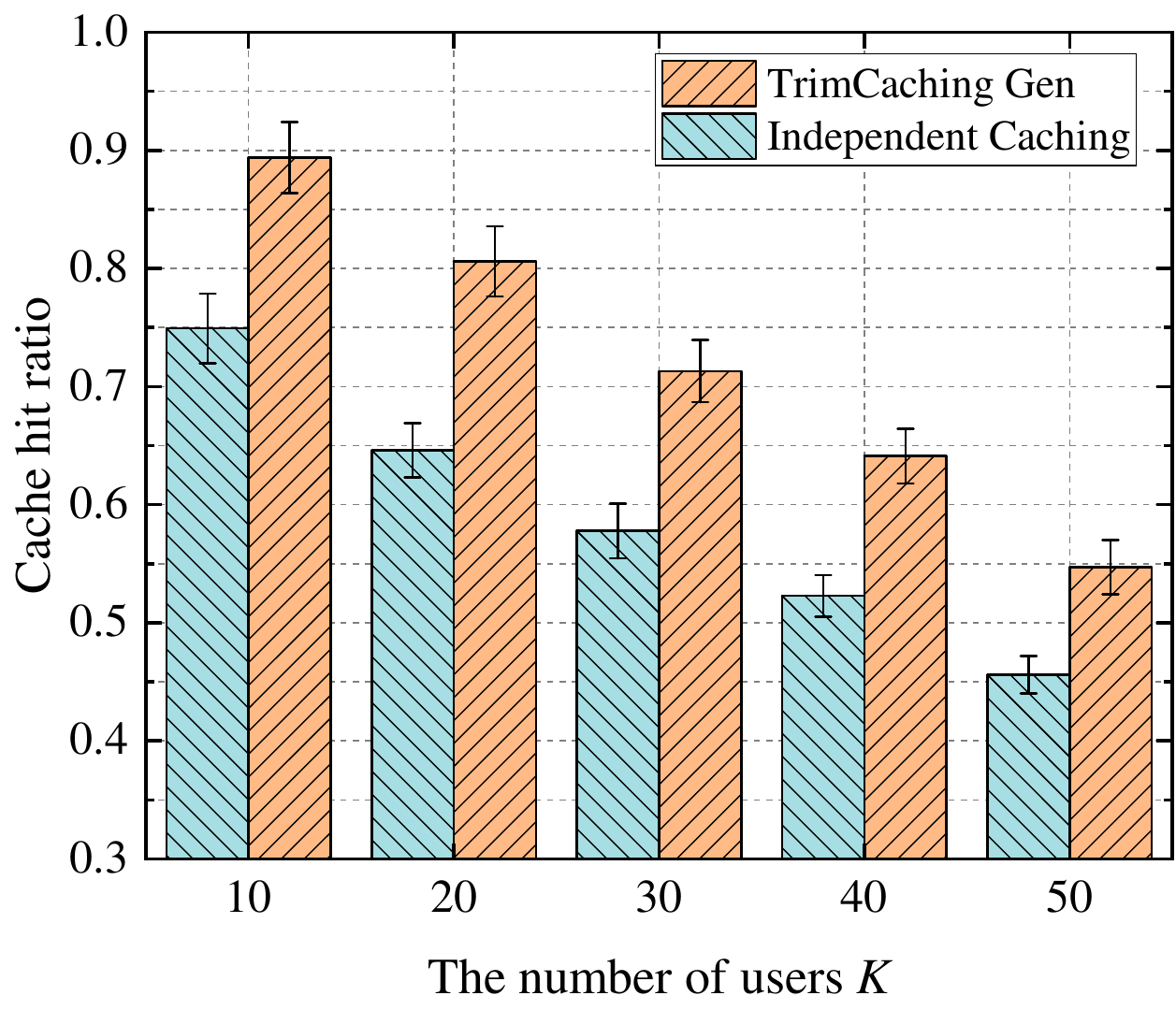}\label{fig_real_case3_user}}
	\quad
	\caption{Cache hit ratio for the general case.}\label{fig_real_case3}
\end{figure*}


In Fig. \ref{fig_real_case3}, we compare the cache hit ratio of the TrimCaching Gen algorithm and the Independent Caching framework in the general case, where parameter blocks can be shared arbitrarily. A trend similar to Fig. \ref{fig_special_case} can be observed. When increasing the storage capacity or the number of edge servers, the cache hit ratio increases because more AI models can be cached. When increasing the number of users, the cache hit ratio decreases due to the limited spectrum bandwidth. More importantly, the TrimCaching Gen algorithm significantly outperforms the Independent Caching framework in different scenarios. 


\subsection{Running Time Comparison}
In this subsection, we compare the running time among the TrimCaching Spec algorithm, the TrimCaching Gen algorithm, and the exhaustive search. Here, the exhaustive search is used to obtain the optimal solution, and the complexity of the exhaustive search is exponential, i.e., $2^{MKI}$. To run the exhaustive search efficiently, the length and width of the area are reduced to 400 m while $M$ and $K$ are set to 2 and 6. $\epsilon$ of Algorithm \ref{algorithm_DP} is set as 0 in this subsection. 

\begin{figure}[!h]
	\centering
        \subfigure[Different algorithms in the special case, where $Q=0.1\ \text{GB}$ and each user requests 9 models.]{\includegraphics[width=1.5in]{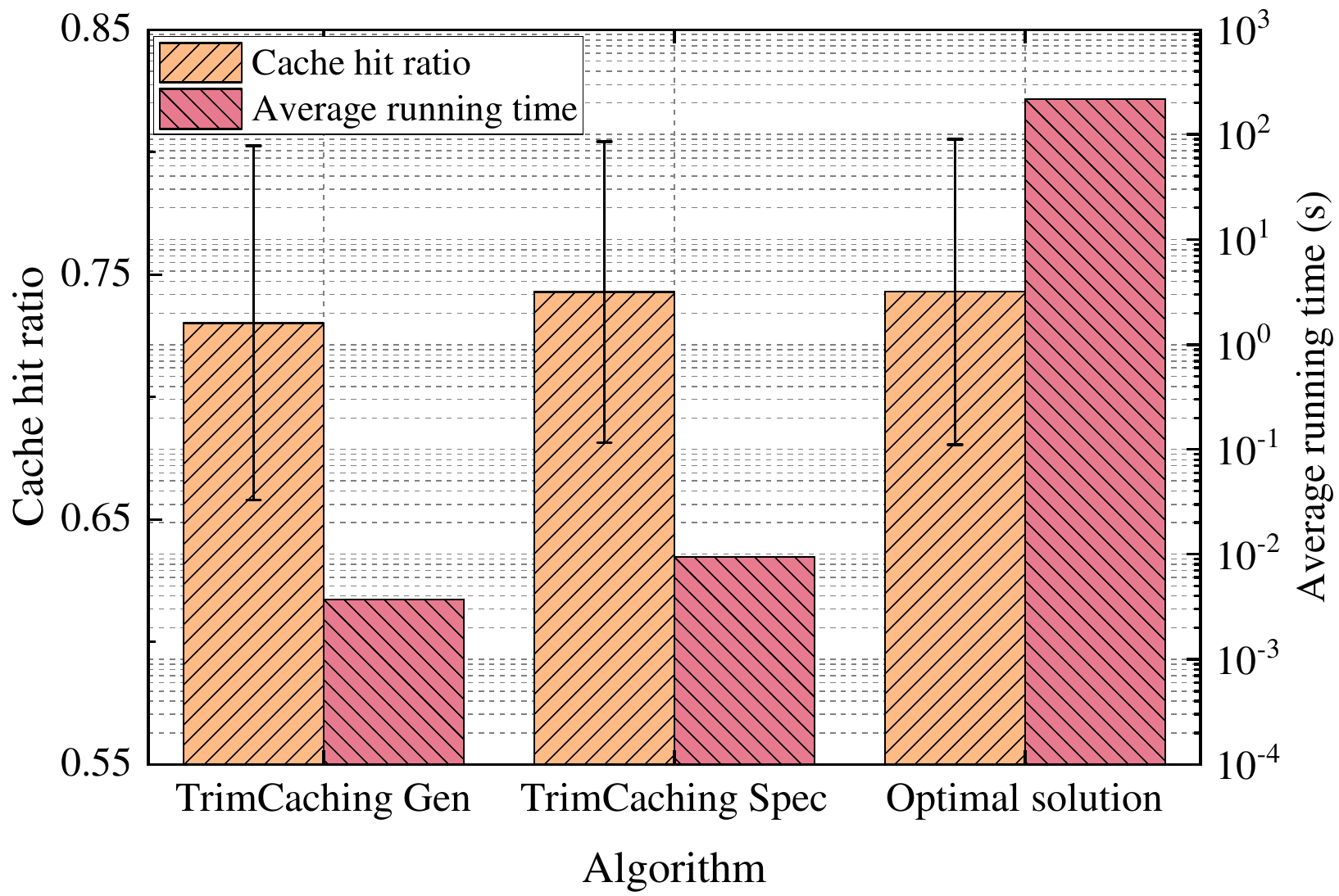}\label{fig_case2_optimal}}
	\quad
	\subfigure[Different algorithms in the general case, where $Q=0.2\ \text{GB}$ and each user requests 27 models.]{\includegraphics[width=1.5in]{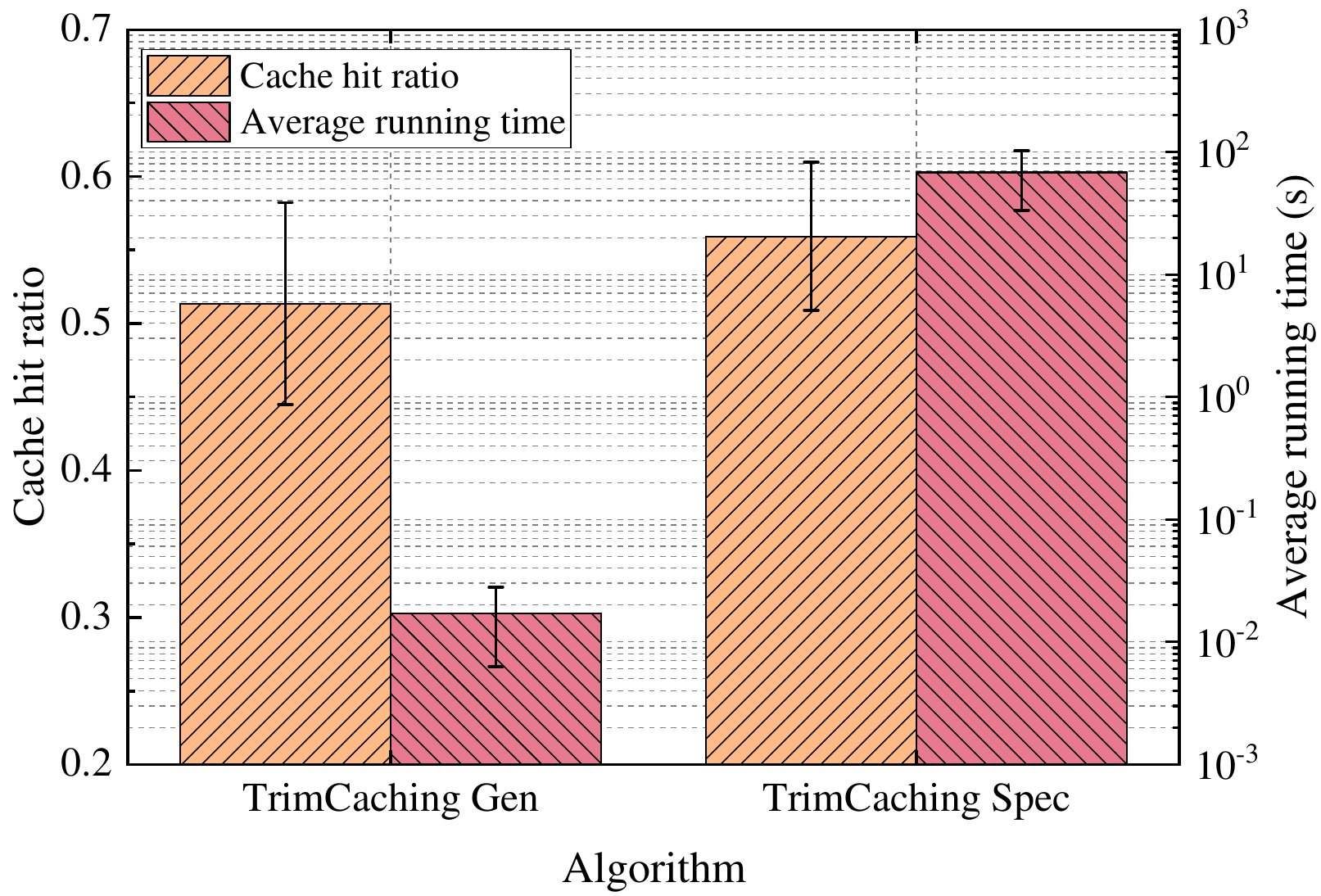}\label{fig_case2_case3}}
 \caption{The cache hit ratio and average running time of different algorithms.}
\end{figure}\label{fig_optimal}
\par 
Fig. \ref{fig_case2_optimal} compares the performance of our algorithms with the exhaustive search in the special case. On the one hand, the cache hit ratio of the TrimCaching Spec algorithm achieves the same cache hit ratio as that of the optimal solution. Moreover, the cache hit ratio of the TrimCaching Gen algorithm is only 1.3\% lower than the optimal solution. On the other hand, the TrimCaching Spec algorithm and the TrimCaching Gen algorithm are, on average, about 22,900 times and 58,000 times faster than the exhaustive search. This phenomenon demonstrates the effectiveness and efficiency of our algorithms in the special case. However, when applied to the general case, the TrimCaching Spec algorithm has exponential complexity. Fig. \ref{fig_case2_case3} compares the performance of the TrimCaching Gen algorithm with the TrimCaching Spec algorithm in the general case. The average running time of the TrimCaching Gen algorithm is about 3,900 times faster than the TrimCaching Spec algorithm, which illustrates the necessity of designing the TrimCaching Gen algorithm for the general case.
\subsection{Robustness to User Mobility}
In Fig. \ref{fig_mobility}, we demonstrate the robustness of our algorithms versus time by considering user mobility. $M$, $K$, and $Q$ are set as 10, 10, and 1 GB. We set three kinds of mobility patterns for pedestrians, bikes, and vehicles. The initial speeds of users are randomly generated from $\left[0.5,1.8\right]\ \text{m/s}$, $\left[2,8\right]\ \text{m/s}$, and $\left[5.5,20\right]\ \text{m/s}$. In each time slot, the acceleration is generated from $\left[-0.3,0.3\right]\ \text{m/}\text{s}^\text{2}$, $\left[-1,1\right]\ \text{m/}\text{s}^\text{2}$, and $\left[-3,3\right]\ \text{m/}\text{s}^\text{2}$. The initial orientations are uniformly distributed in $\left[0,\pi\right]\ \text{rad}$. The angular velocities range in $\left[-\frac{\pi}{4},\frac{\pi}{4}\right]\ \text{rad/s}$, $\left[-\frac{\pi}{3},\frac{\pi}{3}\right]\ \text{rad/s}$, and $\left[-\frac{\pi}{2},\frac{\pi}{2}\right]\ \text{rad/s}$. Users are assumed to change their speeds/orientations at the beginning of each time slot, and the length of each time slot is 5 s. \par 

\begin{figure}[h]
\centerline{\includegraphics[width=2in]{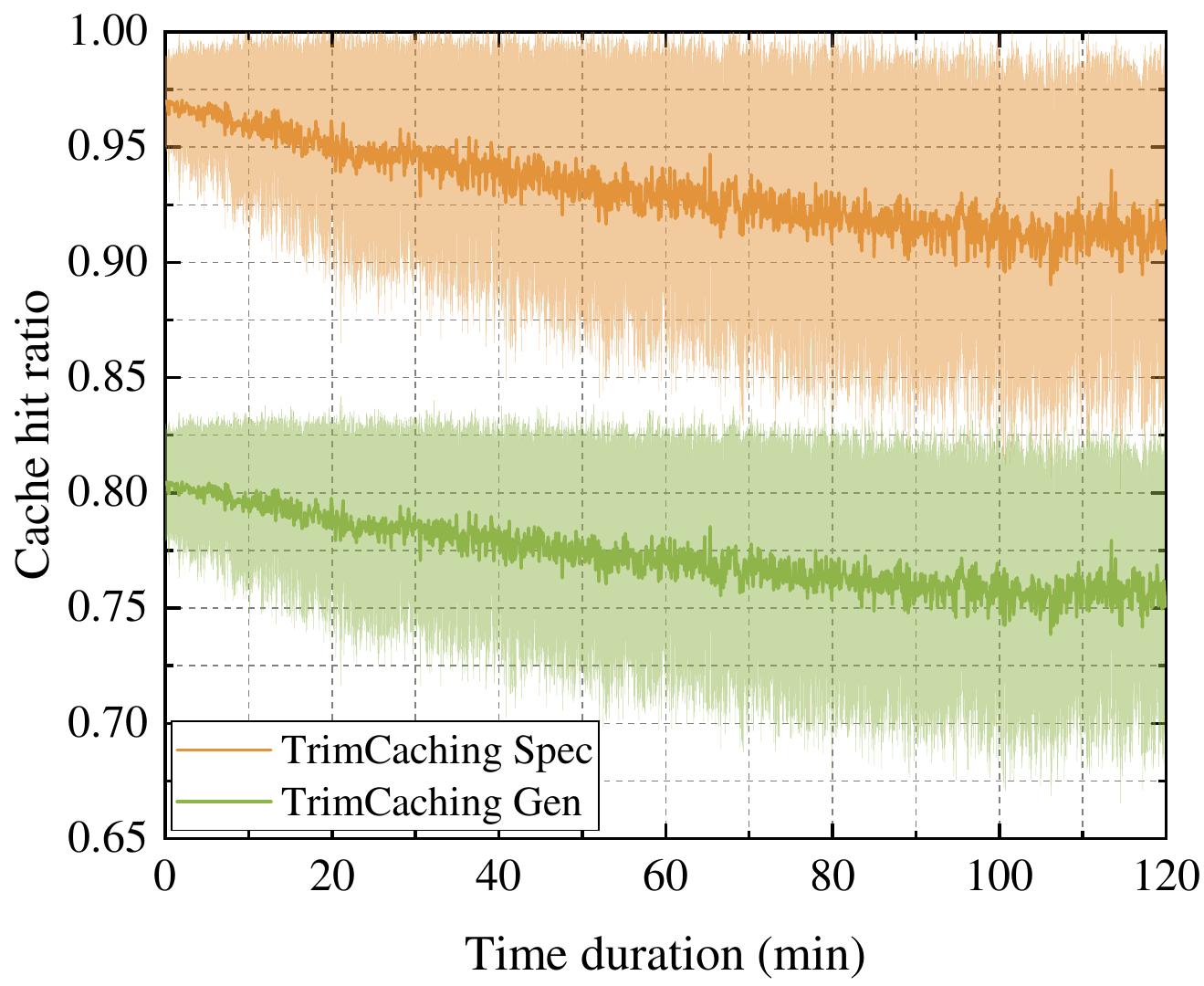}}
	\caption{The cache hit ratio varies over time due to user mobility.
 }
	\label{fig_mobility}
\end{figure}
Fig. \ref{fig_mobility} shows the robustness of the TrimCaching Spec algorithm and the TrimCaching Gen algorithm in the special case. While the cache hit ratio degrades over time due to user mobility, the performance of the TrimCaching Spec algorithm and the TrimCaching Gen algorithm only degrades by about 6.43\% and 5.42\% over 2 h. This shows the effectiveness of the TrimCaching framework in the long run even though user mobility has not been explicitly considered in our problem formulation. This also implies that model replacement does not need to be re-conducted frequently, thereby saving backbone bandwidth resources. 

\section{Conclusions}
In this paper, we proposed the TrimCaching framework for edge AI model caching. Specifically, by taking advantage of the shared parameter blocks among AI models for efficient storage, we investigated the parameter-sharing AI model placement problem. Considering multi-edge scenarios, the model placement problem was formulated to maximize the cache hit ratio under the constraints of service latency requirements and storage capacity. Then, we identified this problem as a submodular maximization problem with submodular constraints, revealing that no polynomial-time algorithm can solve it with a constant approximation ratio. To tackle this challenge, we first investigated the special case with a small fixed number of shared parameter blocks independent of the problem scale. We developed a polynomial-time algorithm with $\left(1-\epsilon\right)/2$-approximation guarantee. After that, we addressed the general case by devising a greedy algorithm that is also efficient and effective. It has been shown that the proposed TrimCaching framework significantly improves the cache hit ratio compared with traditional content placement strategies without considering shared parameter blocks among AI models. We hope this work will inspire further research on edge AI model caching, which has the potential to become a vital component in 6G edge intelligence.
\bibliographystyle{IEEEtran}
\bibliography{IEEEabrv,reference}
\end{document}